\documentclass[11pt]{article}

\usepackage[utf8]{inputenc}
\usepackage[T1]{fontenc}
\usepackage[UKenglish]{babel}

\usepackage{amsmath}
\usepackage{amssymb}
\usepackage{amsthm}
\usepackage[a4paper,top=25mm,left=23mm,right=23mm,bottom=25mm]{geometry}
\usepackage{graphicx}
\usepackage[dvipsnames,svgnames]{xcolor}
\usepackage{hyperref}
\usepackage{thm-restate}
\usepackage{authblk}
\usepackage{enumitem}
\usepackage{mathtools}
\usepackage{colortbl}
\usepackage{textgreek}
\usepackage{tikz}
\usetikzlibrary{patterns}
\usepackage{xspace}
\usepackage[
noline,noend,linesnumbered]{algorithm2e}
\usepackage{ragged2e}

\author[1,2]{Elahe Ghasemi}
\author[2]{Vincent Jugé}
\author[1,3]{Ghazal Khalighinejad}
\author[1,2]{Helia Yazdanyar}
\affil[1]{Sharif University of Technology, Iran}
\affil[2]{LIGM, Université Gustave Eiffel \& CNRS, France}
\affil[3]{Duke University, United States of America}

\title{Galloping in fast-growth natural merge sorts}

\date{}

\newcommand{\NMS}{NaturalMergeSort\xspace}
\newcommand{\MS}{MergeSort\xspace}
\newcommand{\TS}{TimSort\xspace}
\newcommand{\aMS}{\textalpha-MergeSort\xspace}
\newcommand{\aSS}{\textalpha-StackSort\xspace}
\newcommand{\PoS}{PowerSort\xspace}
\newcommand{\PeS}{PeekSort\xspace}
\newcommand{\ShS}{ShiversSort\xspace}
\newcommand{\cASS}{adaptive ShiversSort\xspace}
\newcommand{\CASS}{Adaptive ShiversSort\xspace}
\newcommand{\QSS}{QuickSynergySort\xspace}

\newlength{\lenIf}
\newlength{\lenElse}
\newlength{\lenElseIf}
\newlength{\lengthA}
\newlength{\lengthB}
\newcommand\muteelse[1]{%
 \settowidth{\lengthA}{\textrm{#1}}%
 \hspace{\lenIf-\lenElse+\lengthA}%
}
\newcommand\muteif[2]{%
 \settowidth{\lengthA}{\textrm{#1}}%
 \settowidth{\lengthB}{\textrm{#2}}%
 \hspace{\lenElseIf-\lenIf-\lengthA+\lengthB}%
}
\newcommand\muteelseif[2]{%
 \settowidth{\lengthA}{\textrm{#1}}%
 \settowidth{\lengthB}{\textrm{#2}}%
 \hspace{\lenIf-\lenElseIf+\lengthA-\lengthB}%
}
\newcommand\elseifmute[1]{%
 \settowidth{\lengthA}{\textrm{#1}}%
 \hspace{\lenElseIf-\lenElse+\lengthA}%
}
\newcommand\muteself[2]{%
 \settowidth{\lengthA}{\textrm{#1}}%
 \settowidth{\lengthB}{\textrm{#2}}%
 \hspace{\lengthB-\lengthA}%
}

\newcommand{\A}{\mathcal{A}}
\newcommand{\C}{\mathsf{C}}

\renewcommand{\H}{\mathcal{H}}
\renewcommand{\O}{\mathcal{O}}
\newcommand{\R}{\mathcal{R}}
\renewcommand{\S}{\mathcal{S}}
\newcommand{\T}{\mathcal{T}}
\newcommand{\X}{\mathcal{X}}
\newcommand{\Y}{\mathcal{Y}}
\newcommand{\Z}{\mathbb{Z}}

\newcommand{\ovR}{\overline{R}}
\newcommand{\ovr}{\overline{r}}

\newcommand{\ovh}{\overline{h}}
\newcommand{\aR}[1]{R^{(#1)}}
\newcommand{\ar}[1]{r^{(#1)}}
\newcommand{\ap}[1]{p^{(#1)}}
\newcommand{\al}[1]{\ell^{(#1)}}
\newcommand{\aS}[1]{S^{\langle #1 \rangle}}

\newcommand{\true}{\textbf{true}\xspace}

\newcommand{\cost}{\mathsf{cost}}

\newcommand{\intt}{\mathbf{I}}
\newcommand{\bt}{\mathbf{t}}
\newcommand{\btinit}{\mathbf{t}_{\mathsf{init}}}
\newcommand{\gr}{\mathsf{gr}}
\newcommand{\sle}{\mathsf{sl}}

\newcommand{\ms}{\mathsf{ms}}

\newcommand{\sL}{\mathsf{L}}
\newcommand{\sR}{\mathsf{R}}

\newcommand{\bU}{\mathbf{U}}
\newcommand{\bX}{\mathbf{X}}

\newcommand{\hurl}[4]{\href{https://#1#2#4}{\texttt{\detokenize{#1#3#4}}}}
\newcommand{\splithurl}[5]{\hfill\href{https://#1#2#3#5}{\texttt{\detokenize{#1}}\newline\texttt{\detokenize{#2#4#5}}}}

\renewcommand{\gets}{\ensuremath{\leftarrow}}

\let\oldnl\nl
\newcommand{\drawline}{\BlankLine\renewcommand{\nl}{\let\nl\oldnl}\hrulefill\BlankLine\setcounter{AlgoLine}{0}}
\SetKwIF{If}{ElseIf}{Else}{if}{\!:}{else if}{else:}{endif}
\SetKwFor{For}{for}{\!:}{endf }
\SetKwFor{While}{while}{\!:}{endw}
\SetKwComment{tcp}{ {$\rhd$} }{}
\SetCommentSty{textrm}
\SetInd{0.5em}{0.5em}

\SetKwInput{Input}{\rlap{Input}\phantom{Result\,}}
\SetKwInput{Result}{Result\,}
\SetKwInput{Note}{\rlap{Note}\phantom{Result\,}}
\SetKwProg{Fun}{Function}{:}{}

\newcolumntype{L}[1]{>{\raggedright\let\newline\\\arraybackslash\hspace{0pt}}m{#1}}
\newcolumntype{C}[1]{>{\centering\let\newline\\\arraybackslash\hspace{0pt}}m{#1}}
\newcolumntype{R}[1]{>{\raggedleft\let\newline\\\arraybackslash\hspace{0pt}}m{#1}}

\newtheorem{theorem}{Theorem}
\newtheorem{lemma}[theorem]{Lemma}
\newtheorem{corollary}[theorem]{Corollary}
\newtheorem{proposition}[theorem]{Proposition}
\theoremstyle{definition}
\newtheorem{definition}[theorem]{Definition}

\begin{document}

\maketitle

\begin{abstract}
We study the impact of merging routines in merge-based sorting algorithms.
More precisely, we focus on the \emph{galloping} routine that \TS uses to merge monotonic sub-arrays, hereafter called \emph{runs}, and on the impact on the number of element comparisons performed if one uses this routine instead of a naïve merging routine.

This routine was introduced in order to make \TS more efficient on arrays with few distinct values.
Alas, we prove that, although it makes \TS sort array with two values in linear time, it does not prevent \TS from requiring up to~$\Theta(n \log(n))$ element comparisons to sort arrays of length~$n$ with three distinct values.
However, we also prove that slightly modifying \TS's galloping routine results in requiring only~$\O(n + n \log(\sigma))$ element comparisons in the worst case, when sorting arrays of length~$n$ with~$\sigma$ distinct values.

We do so by focusing on the notion of \emph{dual runs}, which was introduced in the 1990s, and on the associated \emph{dual run-length entropy}.
This notion is both related to the number of distinct values and to the number of runs in an array, which came with its own \emph{run-length entropy} that was used to explain \TS's otherwise ``supernatural'' efficiency.
We also introduce new notions of \emph{fast-} and \emph{middle-growth} for natural merge sorts (i.e., algorithms based on merging runs), which are found in several merge sorting algorithms similar to \TS.

We prove that algorithms with the fast- or middle-growth property, provided that they use our variant of \TS's galloping routine for merging runs, are as efficient as possible at sorting arrays with low run-induced or dual-run-induced complexities.
\end{abstract}


\section{Introduction}\label{sec:intro}

In 2002, Tim Peters, a software engineer, created a new sorting algorithm, which was called \TS~\cite{Peters2015} and was built on ideas from McIlroy~\cite{McIlroy1993}.
This algorithm immediately demonstrated its efficiency for sorting actual data, and was adopted as the standard sorting algorithm in core libraries of widespread programming languages such as Python and Java.
Hence, the prominence of such a custom-made algorithm over previously preferred \emph{optimal} algorithms contributed to the regain of interest in the study of sorting algorithms.

Among the best-identified reasons behind the success of \TS lies the fact that this algorithm is well adapted to the architecture of computers (e.g., for dealing with cache issues) and to realistic distributions of data.
In particular, the very conception of \TS makes it particularly well-suited to sorting data whose \emph{run decompositions}~\cite{BaNa13,EsCaWo92} (see Figure~\ref{fig:runs}) are simple.
Such decompositions were already used in 1973 by Knuth's \NMS~\cite[Section 5.2.4]{Knuth98}, which adapted the traditional \MS algorithm as follows:
\NMS is based on splitting arrays into monotonic subsequences, also called \emph{runs}, and on merging these runs together.
All algorithms sharing this feature of \NMS are also called \emph{natural} merge sorts.

\begin{figure}[t]
\centerline{$
S=(\,\,\underbrace{12,7,6,5}_{\text{first run}},
\,\,\underbrace{5,7,14,36}_{\text{second run}},
\,\,\underbrace{3,3,5,21,21}_{\text{third run}},
\,\,\underbrace{20,8,5,1}_{\text{fourth run}}\,\,)~$}
\caption{A sequence and its \emph{run decomposition} computed by a greedy algorithm:
for each run, the first two elements determine if the run is non-decreasing or decreasing, and the run continues with the maximum number of consecutive elements that preserve its monotonicity.
\label{fig:runs}} 
\end{figure}

In addition to being a natural merge sort, \TS includes many optimisations, which were carefully engineered, through extensive testing, to offer the best complexity performances.
As a result,  the general structure of \TS can be split into three main components:
(i) a variant of an insertion sort, used to deal with \emph{small} runs, e.g., runs of length less than 32,
(ii) a simple policy for choosing which \emph{large} runs to merge,
(iii) a routine for merging these runs, based on a so-called \emph{galloping} strategy.
The second component has been subject to an intense scrutiny these last few years, thereby giving birth to a great variety of {\TS}-like algorithms, such as \mbox{\aSS}~\cite{AuNiPi15}, \mbox{\aMS}~\cite{BuKno18},  \ShS~\cite{shivers02} (which \emph{predated} \TS), \cASS~\cite{Ju20}, \PeS and \PoS~\cite{munro2018nearly}.\label{citesorts}
On the contrary, the first and third components, which seem more complicated and whose effect might be harder to quantify, have often been used as black boxes when studying \TS or designing variants thereof.

In what follows, we focus on the third component and prove that it can be made very efficient:
although \TS may require up to~$\Theta(n \log(n))$ comparisons to sort arrays of length~$n$ with three distinct values, slight modifications to the galloping routine make \TS require only~$\O(n+n \log(\sigma))$ comparisons to sort arrays of length~$n$ with~$\sigma$ distinct values.
This is reminiscent of the celebrated complexity result~\cite{auger2018worst} stating that \TS requires~$\O(n+n \log(\rho))$ comparisons to sort arrays of length~$n$ that can be decomposed as a concatenation of~$\rho$ non-decreasing arrays.

\paragraph*{Context and related work.}

The success of \TS has nurtured the interest in the quest for sorting algorithms that would be both excellent all-around and adapted to arrays with few runs.
However, its \emph{ad hoc} conception made its complexity analysis harder than what one might have hoped, and it is only in 2015, a decade after \TS had been largely deployed, that Auger et al.~\cite{AuNiPi15} proved that \TS required~$\O(n \log(n))$ comparisons for sorting arrays of length~$n$.

This is optimal in the model of sorting by comparisons, if the input array can be an arbitrary array of length~$n$.
However, taking into account the run decompositions of the input array allows using finer-grained complexity classes, as follows.
First, one may consider only arrays whose run decomposition consists of~$\rho$ monotonic runs.
On such arrays, the best worst-case time complexity one may hope for is~$\O(n + n \log(\rho))$~\cite{Mannila1985}.
Second, we may consider even more restricted classes of input arrays, and focus only on those arrays that consist of~$\rho$ runs of lengths~$r_1,\ldots,r_\rho$.
In that case, every comparison-based sorting algorithm requires at least~$n \H + \O(n)$ element comparisons, where~$\H$ is defined as~$\H = H(r_1/n,\ldots,r_\rho/n)$ and~$H(x_1,\ldots,x_\rho) = -\sum_{i=1}^\rho x_i \log_2(x_i)$ is the general entropy function~\cite{BaNa13,Ju20,McIlroy1993}.
The number~$\H$ is called the \emph{run-length entropy} of the array.

Since the early 2000s, several natural merge sorts were proposed, all of which were meant to offer easy-to-prove complexity guarantees:
\ShS, which runs in time~$\O(n \log(n))$;
\aSS, which, like \NMS, runs in time~$\O(n + n \log(\rho))$;
\mbox{\aMS}, which, like \TS, runs in time~$\O(n + n \H)$;
\cASS, \PeS and \PoS, require only~$n \H + \O(n)$ comparisons and element moves.

Except \TS, these algorithms are, in fact, described only as policies for deciding which runs to merge, the actual routine used for merging runs being left implicit:
since choosing a naïve merging routine does not harm the worst-case time complexities considered above, all authors identified the cost of merging two runs of lengths~$m$ and~$n$ with the sum~$m+n$, and the complexity of the algorithm with the sum of the costs of the merges performed.

One notable exception is that of Munro and Wild~\cite{munro2018nearly}.
They compared the running times of \TS and of \TS's variant obtained by using a naïve merging routine instead of \TS's galloping routine.
However, and although they mentioned the challenge of finding distributions on arrays that might benefit from galloping, they did not address this challenge, and focused only on arrays with a low entropy~$\H$.
As a result, they unsurprisingly observed that the galloping routine looked \emph{slower} than the naïve one.

Galloping turns out to be very efficient when sorting arrays with few distinct values, a class of arrays that had also been intensively studied.
As soon as 1976, Munro and Spira~\cite{MuSpi76} proposed a complexity measure~$\H^\ast$ related to the run-length entropy, with the property that~$\H^\ast \leqslant \log_2(\sigma)$ for arrays with~$\sigma$ values.
They also proposed an algorithm for sorting arrays of length~$n$ with~$\sigma$ values by using~$\O(n + n \H^\ast)$ comparisons.
McIlroy~\cite{McIlroy1993} then extended their work to arrays representing a permutation~$\pi$, identifying~$\H^\ast$ with the run-length entropy of~$\pi^{-1}$ and proposing a variant of Munro and Spira's algorithm that would use~$\O(n + n \H^\ast)$ comparisons in this generalised setting.
Similarly, Barbay et al.~\cite{BaOchoaSatti16} invented the algorithm \QSS, which aimed at minimising the number of comparisons, achieving a~$\O(n + n \H^\ast)$ upper bound and further refining the parameters it used, by taking into account the interleaving between runs and dual runs.
Yet, all of these algorithms require~$\omega(n + n \H)$ element moves in the worst case.

Furthermore, as a side effect of being rather complicated and lacking a proper analysis, except that of~\cite{munro2018nearly} that hinted at its inefficiency, the galloping routine has been omitted in various mainstream implementations of natural merge sorts, in which it was replaced by its naïve variant.
This is the case, for instance, in library \TS implementations of the programming languages Swift~\cite{Swift} and Rust~\cite{Rust}.
On the contrary, \TS implementations in other languages, such as Java~\cite{Java}, Octave~\cite{Octave} or the V8 JavaScript engine~\cite{V8}, and \PoS implementation in Python~\cite{Python}
include the galloping routine.

\paragraph*{Contributions.}

We study the time complexity of various natural merge sort algorithms in a context where arrays are not just parametrised by their lengths.
More precisely, we focus on a decomposition of input arrays that is dual to the decomposition of arrays into monotonic runs, and that was proposed by McIlroy~\cite{McIlroy1993}.

Consider an array~$A$ that we want to sort in a \emph{stable} manner, i.e., in which two elements can always considered to be distinct, if only because their positions in~$A$ are distinct:
without loss of generality, we identify the values~$A[1],A[2],\ldots,A[n]$ with the integers from~$1$ to~$n$, thereby making~$A$ a permutation of the set~$\{1,2,\ldots,n\}$.
A common measure of presortedness consists in subdividing~$A$ into distinct monotonic \emph{runs}, i.e., partitioning the set~$\{1,2,\ldots,n\}$ into intervals~$R_1,R_2,\ldots,R_\rho$ on which the function~$x \mapsto A[x]$ is monotonic.
Although this partition is \emph{a priori} not unique, there exists a simple greedy algorithm that provides us with a partition for which~$\rho$ is minimal:
it suffices to construct the intervals~$R_1,R_2,\ldots,R_\rho$ from left to right, each time choosing~$R_i$ to be as long as possible.

\begin{figure}[t]
\begin{center}
\begin{tikzpicture}[scale=0.32]
\foreach \x in {26}{
\foreach \a/\b in {1/5,9/12,15/16}{
 \draw[draw=black!20,fill=black!20] (0.5,\a-0.5) -- (0.5,\b+0.5) -- (18.5,\b+0.5) -- (18.5,\a-0.5) -- cycle;
 \draw[draw=black!20,fill=black!20] (\x+\a-0.5,0.5) -- (\x+\b+0.5,0.5) -- (\x+\b+0.5,18.5) -- (\x+\a-0.5,18.5) -- cycle;
}
}

\foreach \x in {0,26}{
 \foreach \j in {1,...,18} {
  \draw (\x+0.5,\j) -- (\x+18.5,\j) (\x+\j,0.5) -- (\x+\j,18.5);
  \node[anchor=east] at (\x+0.5,\j) {\tiny \j};
  \node[anchor=north] at (\x+\j,0.5) {\tiny \j};
 }
 \draw[gray] (\x+0.5,0.5) -- (\x+18.5,0.5) -- (\x+18.5,18.5) -- (\x+0.5,18.5) -- cycle;
}

\foreach \x in {26}{
\foreach \i/\c [count=\j] in {1/black,3/black,7/black,10/black,15/black,14/white,6/white,2/white,4/black,9/black,11/black,16/black,5/white,17/white,13/black,8/black,12/white,18/white} {
 \draw[draw=black,fill=\c] (\i,\j) circle (0.25);
 \draw[draw=black,fill=\c] (\x+\j,\i) circle (0.25);
}
}
\node[anchor=south] at (9.5,18.75) {Permutation~$\pi$};
\node[anchor=south] at (26+9.5,18.75) {Permutation~$\pi^{-1}$};

\begin{scope}[shift={(13,22.5)}]
\draw[draw=black!20,fill=black!20]
(0.5,-0.5) -- (0.5,2.5) -- (18.5,2.5) -- (18.5,-0.5) -- cycle
(0.5,5.5) -- (2.5,5.5) -- (3.5,4.5) -- (18.5,4.5) -- (18.5,6.5) -- (0.5,6.5) -- cycle
(0.5,7.5) -- (0.5,9.5) -- (18.5,9.5) -- (18.5,7.5) -- cycle;

\foreach \j in {0,...,12}{
 \draw (0.5,\j) -- (18.5,\j);
 \node[anchor=east] at (0.5,\j) {\tiny \j};
}
\foreach \i in {1,...,18}{
 \draw (\i,-0.5) -- (\i,12.5);
 \node[anchor=north] at (\i,-0.5) {\tiny \i};
}
\draw[gray] (0.5,-0.5) -- (18.5,-0.5) -- (18.5,12.5) -- (0.5,12.5) -- cycle;

\foreach \j/\c [count=\i] in {0/black,5/white,0/black,5/black,7/white,4/white,0/black,9/black,5/black,2/black,6/black,10/white,8/black,3/white,2/black,6/black,7/white,12/white}{
 \draw[draw=black,fill=\c] (\i,\j) circle (0.25);
}
\node[anchor=south] at (9.5,12.75) {Array~$A$};
\end{scope}
\end{tikzpicture}
\vspace{-1.7em}
\end{center}
\caption{The array~$A$ is lexicographically equivalent to the permutation~$\pi$.
Its dual runs, represented with gray and white horizontal stripes, have respective lengths~$5$,~$3$,~$4$,~$2$,~$2$ and~$3$.
The mapping~$A \mapsto \tau$ identifies them with the dual runs of~$\tau$, i.e., with the runs of the permutation~$\tau^{-1}$.
Note that, although they are on the same horizontal line in the diagrammatic representation of~$A$, the points with coordinates~$(2,5)$ and~$(4,5)$ belong to distinct dual runs.\label{fig:dual-runs}}
\end{figure}

Here, we adopt a dual approach.
We first start by identifying~$A$ with the permutation it is lexicographically equivalent to.
Then, we partition the set~$\{1,2,\ldots,n\}$ into the monotonic\footnote{In an earlier version of this article, published at the conference ICALP, we required the runs~$S_i$ to be increasing.
We no longer do so, and allow both increasing and decreasing runs.} runs~$S_1,S_2,\ldots,S_\sigma$ of the inverse permutation~$A^{-1}$.
These intervals~$S_i$ are already known under the name of \emph{riffle shuffles}~\cite{McIlroy1993}.
In order to underline their connection with runs, we say that these intervals are the \emph{dual runs} of~$A$, and we denote their lengths by~$s_i$.
The process of transforming an array into a permutation and then extracting its dual runs is illustrated in Figure~\ref{fig:dual-runs}.

In particular, if an array has~$\sigma$ values, it cannot have more than~$\sigma$ dual runs.
Note, however, that it may have significantly fewer than~$\sigma$ dual runs, as shown by the examples of the monotonic permutations, which have~$n$ values but only one dual run.
In addition, in general, there is no non-trivial connection between the runs of a permutation and its dual runs.
For instance, a permutation with a given number of runs may have arbitrarily many (or few) dual runs, and conversely.

In this article, we prove that, by using a variant of \TS's galloping routine, several natural merge sorts require only~$\O(n + n \H^\ast)$ comparisons, or even~$n \H^\ast + \O(n \log(\H^\ast+2))$ comparisons, where~$\H^\ast = H(s_1/n,\ldots,s_\sigma/n) \leqslant \log_2(\sigma)$ is called the \emph{dual run-length entropy} of the array,~$s_i$ is the length of the dual run~$S_i$, and~$H$ is the general entropy function already mentioned above.
We also prove that, unfortunately, using \TS's galloping routine verbatim does \emph{not} provide such complexity guarantees, since \TS itself requires~$\Theta(n \log(n))$ comparisons to sort arrays of length~$n$ with only three distinct values, in the worst case.

This legitimates using our variant of \TS's arguably complicated galloping routine rather than its naïve variant, in particular when sorting arrays that are constrained to have relatively few distinct values.

This also subsumes results that have been known since the 1970s.
For instance, adapting the optimal constructions for alphabetic Huffman codes by Hu and Tucker~\cite{HuTucker71} or Garsia and Wachs~\cite{Garsia77} to \emph{merge trees} (described in Section~\ref{sec:fast-growth}) already provided sorting algorithms working in time~$n \H + \O(n)$. 

Our new results rely on notions that we call \emph{fast-} and \emph{middle-growth} properties, and which are found in natural merge sorts like \aMS, \aSS, \cASS, \ShS, \PeS, \PoS or \TS.
More precisely, we prove that merge sorts require~$\O(n + n \H)$ comparisons \emph{and} element moves when they possess the fast-growth property, thereby encompassing complexity results that were proved separately for each of these algorithms~\cite{auger2018worst,BuKno18,Ju20,munro2018nearly}, and~$\O(n + n \H^\ast)$ comparisons when they possess the fast- or middle-growth property, which is a completely new result.

Finally, we prove finer complexity bounds on the number of comparisons used by \cASS, \ShS, \NMS, \PeS and \PoS, which require only~$n \H^\ast + \O(n\log(\H^\ast+2))$ comparisons, nearly matching the~$n \H + \O(n)$ (or~$n \log_2(n) + \O(n)$ and~$n \log_2(\rho) + \O(n)$, in the cases of \ShS and \NMS) complexity upper bound they already enjoy in terms of comparisons and element moves.

\section{The galloping routine for merging runs}
\label{sec:description}

Here, we describe the galloping routine that \TS uses to merge adjacent non-decreasing runs.
This routine is a blend between a naïve merging algorithm, which requires~$a+b-1$ comparisons to merge runs~$A$ and~$B$ of lengths~$a$ and~$b$, and a dichotomy-based algorithm, which requires~$\O(\log(a+b))$ comparisons in the best case, and~$\O(a+b)$ comparisons in the worst case.
It depends on a parameter~$\bt$, and works as follows.

When merging runs~$A$ and~$B$ into one large run~$C$, we first need to find the least integers~$k$ and~$\ell$ such that~$B[0] < A[k] \leqslant B[\ell]$:
the~$k+\ell$ first elements of~$C$ are
\[A[0],A[1],\ldots,A[k-1],B[0],B[1],\ldots,B[\ell-1],\]
and the remaining elements of~$C$ are obtained by merging the sub-array of~$A$ that spans positions~$k$ to~$a$ and the sub-array of~$B$ that spans positions~$\ell$ to~$b$.
Computing~$k$ and~$\ell$ efficiently is therefore a crucial step towards reducing the number of comparisons required by the merging routine (and, thus, by the sorting algorithm).

This computation is a special case of the following problem:
if one wishes to find a secret integer~$m \geqslant 1$ by choosing integers~$x \geqslant 1$ and testing whether~$x \geqslant m$, what is, as a function of~$m$, the least number of tests that one must perform?
Bentley and Yao~\cite{BeYa76} answer this question by providing simple strategies, which they number~$\mathsf{B}_0, \mathsf{B}_1,\ldots$:
\begin{itemize}
\item[$\mathsf{B}_0$:] choose~$x = 1$, then~$x = 2$, and so on, until choosing~$x = m$, thereby finding~$m$ in~$m$ queries.

\item[$\mathsf{B}_1$:] first use~$\mathsf{B}_0$ to find~$\lceil \log_2(m) \rceil+1$ in~$\lceil \log_2(m) \rceil + 1$ queries, i.e., choose~$x = 2^k$ until~$x \geqslant m$, then compute the bits of~$m$ (from the most significant bit of~$m$ to the least significant one) in~$\lceil \log_2(m) \rceil - 1$ additional queries.
For instance, if~$m = 5$, this results in successively choosing~$x = 1, 2, 4, 8, 6$ and~$5$;
if~$m = 8$, this results in successively choosing~$x = 1, 2, 4, 8, 6$ and~$7$.
Bentley and Yao call this strategy a \emph{galloping} (or \emph{exponential search}) technique.

\item[$\mathsf{B}_{k+1}$:] like~$\mathsf{B}_1$, except that one finds~$\lceil \log_2(m) \rceil + 1$ by using~$\mathsf{B}_k$ instead of~$\mathsf{B}_0$.
\end{itemize}

Strategy~$\mathsf{B}_0$ uses~$m$ queries,~$\mathsf{B}_1$ uses~$2 \lceil \log_2(m) \rceil$ queries (except for~$m = 1$, where it uses one query), and each strategy~$\mathsf{B}_k$ with~$k \geqslant 2$ uses~$\log_2(m) + o(\log(m))$ queries.
Thus, if~$m$ is known to be arbitrarily large, one should favour some strategy~$\mathsf{B}_k$ (with~$k \geqslant 1$) over the naïve strategy~$\mathsf{B}_0$.
However, when merging runs taken from a permutation chosen uniformly at random over the~$n!$ permutations of~$\{1,2,\ldots,n\}$, the integer~$m$ is frequently small, which makes~$\mathsf{B}_0$ suddenly more attractive.
In particular, the overhead of using~$\mathsf{B}_1$ instead of~$\mathsf{B}_0$ is a prohibitive~$+20\%$ or~$+33\%$ when~$m = 5$ or~$m = 3$, as illustrated in the black cells of Table~\ref{table:1}.

\begin{table}[b]
\begin{center}
\begin{tabular}{|C{3.75mm}|C{3.35mm}|C{3.35mm}|C{3.35mm}|C{3.35mm}|C{3.35mm}|C{3.35mm}|C{3.35mm}|C{3.35mm}|C{3.35mm}|C{3.35mm}|C{3.35mm}|C{3.35mm}|C{3.35mm}|C{3.35mm}|C{3.35mm}|C{3.35mm}|C{3.35mm}|}
\hline
\clap{$m$} & \clap{$1$} & \clap{$2$} & \clap{$3$} & \clap{$4$} & \clap{$5$} & \clap{$6$} & \clap{$7$} & \clap{$8$} & \clap{$9$} & \clap{$10$} & \clap{$11$} & \clap{$12$} & \clap{$13$} & \clap{$14$} & \clap{$15$} & \clap{$16$} & \clap{$17$} \\
\hline
\clap{$\mathsf{B}_0$} & \clap{$1$} & \clap{$2$} & \clap{$3$} & \clap{$4$} & \clap{$5$} & \clap{$6$} & \clap{$7$} & \clap{$8$} & \clap{$9$} & \clap{$10$} & \clap{$11$} & \clap{$12$} & \clap{$13$} & \clap{$14$} & \clap{$15$} & \clap{$16$} & \clap{$17$} \\
\hline
\clap{$\mathsf{B}_1$} & \clap{$1$} & \clap{$2$} & \clap{\cellcolor{black}\color{white}$\mathbf{4}$} & \clap{$4$} & \clap{\cellcolor{black}\color{white}$\mathbf{6}$} & \clap{$6$} & \clap{$6$} & \clap{$6$} & \clap{$8$} & \clap{$8$} & \clap{$8$} & \clap{$8$} & \clap{$8$} & \clap{$8$} & \clap{$8$} & \clap{$8$} & \clap{$10$} \\
\hline 
\end{tabular}
\end{center}
\caption{Comparison requests needed by
strategies~$\mathsf{B}_0$ and~$\mathsf{B}_1$
to find a secret integer~$m \geqslant 1$.
\label{table:1}}
\end{table}

McIlroy~\cite{McIlroy1993} addresses this issue by choosing a parameter~$\bt$ and using a blend between the strategies~$\mathsf{B}_0$ and~$\mathsf{B}_1$, which consists in two successive steps~$\mathsf{C}_1$ and~$\mathsf{C}_2$:
\begin{enumerate}
\item[$\mathsf{C}_1$:] first follow~$\mathsf{B}_0$ for up to~$\bt$ steps, thereby choosing~$x = 1$,~$x = 2$, \ldots,~$x = \bt$ (if~$m \leqslant \bt-1$, one stops after choosing~$x = m$).
\item[$\mathsf{C}_2$:] if~$m \geqslant \bt+1$, switch to~$ \mathsf{B}_1$ (or, more precisely, to a version of~$\mathsf{B}_1$ translated by~$\bt$, since the precondition~$m \geqslant 1$ is now~$m \geqslant \bt+1$).
\end{enumerate}
Once such a parameter~$\bt$ is fixed, McIlroy's mixed strategy allows retrieving~$m$ in~$\cost_{\bt}(m)$ queries, where~$\cost_{\bt}(m) = m$ if~$m \leqslant \bt+2$, and~$\cost_{\bt}(m) = \bt + 2 \lceil \log_2(m - \bt) \rceil$ if~$m \geqslant \bt+3$.

In practice, however, the integer we are evaluating is not only positive, but even subject to the double inequality~$1 \leqslant k \leqslant a$ or~$1 \leqslant \ell \leqslant b$;
above, we overlooked those upper bounds.
Taking them into account allows us to marginally improve strategy~$\mathsf{B}_1$ and McIlroy's resulting mixed strategy, at the expense of providing us with a more complicated cost function;
in \TS's implementation, this improvement is used only when~$k \geqslant \max\{\bt,a/2\}$ or~$\ell \geqslant \max\{\bt,b/2\}$.

Therefore, and in order too keep things simple, we will replace the precise cost functions we might have obtained by the following simpler upper bound.
By contrast, whenever constructing examples aimed at providing lower bounds, we will make sure that the cases~$k \geqslant \max\{\bt,a/2\}$ and~$\ell \geqslant \max\{\bt,b/2\}$ never occur.

\begin{lemma}\label{lem:cost-bound}
For all~$\bt \geqslant 0$ and~$m \geqslant 1$, we have~$\cost_{\bt}(m) \leqslant \cost_{\bt}^\ast(m)$, where
\[\cost_{\bt}^\ast(m) = \min\{(1 + 1 / (\bt+3)) m, \bt + 2 + 2 \log_2(m + 1)\}.\]
\end{lemma}

\begin{proof}
Since the desired inequality is immediate when~$m \leqslant \bt+2$, we assume that~$m \geqslant \bt+3$.
In that case, we already have~$\cost_{\bt}(m) \leqslant \bt + 2 (\log_2(m-\bt)+1) \leqslant \bt+2+2\log_2(m+1)$, and we prove now that~$\cost_\bt(m) \leqslant m+1$.
Indeed, let~$u = m - \bt$ and let~$f \colon x \mapsto x-1-2\log_2(x)$.
The function~$f$ is positive and increasing on the interval~$[7,+\infty)$.
Thus, it suffices to check by hand that~$(m+1) - \cost_\bt(m) = 0,1,0,1$ when~$u = 3,4,5,6$, and that~$(m+1) - \cost_\bt(m) \geqslant f(u) > 0$ when~$u \geqslant 7$.
It follows, as expected, that~$\cost_{\bt}(m) \leqslant m+1 \leqslant (1 + 1/(\bt+3))m$.
\end{proof}

The above discussion immediately provides us with a cost model for the number of comparisons performed when merging two runs.

\begin{proposition}\label{pro:2}
Let~$\pi$ be a permutation with dual runs~$S_1,S_2,\ldots,S_\sigma$, and let~$A$ and~$B$ be two monotonic runs of lengths~$a$ and~$b$ obtained while sorting~$\pi$ with a natural merge sort.
For each integer~$i \leqslant \sigma$, let~$a_{\rightarrow i}$ (respectively,~$b_{\rightarrow i}$) be the number of elements in~$A$ (respectively, in~$B$) whose value belongs to~$S_i$.
Using a merging routine based on McIlroy's mixed strategy for a fixed parameter~$\bt$, we need at most
\[1 + \sum_{i=1}^\sigma \cost_{\bt}^\ast(a_{\rightarrow i}) + \cost_{\bt}^\ast(b_{\rightarrow i})\]
element comparisons to merge the runs~$A$ and~$B$ into an increasing run.
\end{proposition}

\begin{proof}
Let~$C$ be the increasing run that results from merging~$A$ and~$B$.
A preliminary step consists in checking whether each of~$A$ and~$B$ is increasing or decreasing, in which case they are reversed.
This only requires comparing the first two elements of~$A$ and~$B$.

If~$S_1$ is an increasing run of~$\pi^{-1}$, the run~$C$ starts with those~$a_{\rightarrow 1}$ elements from~$A$ whose value belongs to~$S_1$, and then those~$b_{\rightarrow 1}$ elements from~$B$ whose value belongs to~$S_1$;
otherwise,~$C$ starts with those~$b_{\rightarrow 1}$ elements from~$B$ whose value belongs to~$S_1$, and only then those~$a_{\rightarrow 1}$ elements from~$A$ whose value belongs to~$S_1$.
In both cases,~$C$ then consists of~$a_{\rightarrow 2}+b_{\rightarrow 2}$ elements whose value belongs to~$S_2$ (starting with elements of~$A$ if~$S_2$ is an increasing run of~$\pi^{-1}$, and of~$B$ otherwise), and so on.

Hence, the elements of~$A$ are grouped in blocks of length~$a_{\rightarrow i}$, some of which will form consecutive blocks of elements of~$C$, thereby being ``glued'' together.
Similarly, the elements of~$B$ are grouped in blocks of length~$b_{\rightarrow i}$, some being glued together.
The run~$C$ will then consist in an alternation of (possibly, glued) blocks from~$A$ and~$B$, and the galloping routine consists in discovering whether the first block comes from~$A$ or from~$B$ (which takes one query) and then successively computing the lengths of these blocks, except the last one (because once a run has been completely integrated into~$C$, the remaining elements of the other run will necessarily form one unique block).

Since~$\cost_{\bt}^\ast$ is sub-additive, i.e., since~$\cost_{\bt}^\ast(m) + \cost_{\bt}^\ast(m') \geqslant \cost_{\bt}^\ast(m+m')$ for all~$m \geqslant 0$ and~$m' \geqslant 0$, discovering the length of glued block of length~$a_{\rightarrow i}+a_{\rightarrow i+1}+\cdots+a_{\rightarrow j}$ requires no more than
\[\cost_{\bt}^\ast(a_{\rightarrow i}+a_{\rightarrow i+1}+\cdots+a_{\rightarrow j}) \leqslant \cost_{\bt}^\ast(a_{\rightarrow i})+\cost_{\bt}^\ast(a_{\rightarrow i+1})+\cdots+\cost_{\bt}^\ast(a_{\rightarrow j})\]
element comparisons.
Using this upper bound to count comparisons of elements taken from all blocks (glued or not) that belong to either~$A$ or~$B$ completes the proof.
\end{proof}

We simply call~$\bt$-\emph{galloping routine} the merging routine based on McIlroy's mixed strategy for a fixed parameter~$\bt$;
when the value of~$\bt$ is irrelevant, we simply omit mentioning it.
Then, the quantity
\[1 + \sum_{i=1}^\sigma \cost_{\bt}^\ast(a_i) +
\cost_{\bt}^\ast(b_i)\]
is called the ($\bt$-)\emph{galloping cost} of merging~$A$ and~$B$.
This cost never exceeds~$1 + 1/(\bt+3)$ times the sum~$a + b$, which we call \emph{naïve cost} of merging~$A$ and~$B$.
Below, we study the impact of using the galloping routine instead of the naïve one, which amounts to replacing naïve merge costs by their galloping variants.

Note that using this new galloping cost measure is relevant only if element comparisons are significantly more expensive than element (or pointer) moves.
For example, even if we were lucky enough to observe that each element in~$B$ is smaller than each element in~$A$, we would perform only~$\O(\log(a+b))$ element comparisons, but as many as~$\Theta(a+b)$ element moves.

\paragraph*{Updating the parameter~$\bt$.}

We assumed above that the parameter~$\bt$ did not vary while the runs~$A$ and~$B$ were being merged with each other.
This is not how~$\bt$ behaves in \TS's implementation of the galloping routine.

Instead, the parameter~$\bt$ is initially set to a constant ($\bt = 7$ in Java), and may change during the algorithm, by proceeding roughly as follows.
In step~$\mathsf{C}_2$, after using the strategy~$\mathsf{B}_1$, and depending on the value of~$m$ that we found, one may realise that using~$\mathsf{B}_0$ might have been less expensive than using~$\mathsf{B}_1$.
In that case, the value of~$\bt$ increases by~$1$, and otherwise (i.e., if using~$\mathsf{B}_1$ was indeed a smart move), it decreases by~$1$ (with a minimum of~$0$).
As we will see, however, \TS's actual implementation includes many additional low-level ``optimisations'' compared to this simplified version, a few of which result in \emph{worse} worst-case complexity bounds.

In this paper, we will study and compare three policies for choosing the value of~$\bt$:
\begin{enumerate}
\item setting~$\bt$ to a fixed constant, e.g.,~$\bt = 0$ or~$\bt = 7$;
\item following \TS's update policy;
\item setting~$\bt = \tau \lceil \log_2(a+b) \rceil^2$ whenever we merge runs of lengths~$a$ and~$b$, for some constant~$\tau$;
\end{enumerate}
Since the first policy, which consists in not updating~$\bt$ at all, is the most simple one, it will be our choice by default.
In Section~\ref{sec:fast-growth}, we focus only on generic properties of this ``no-update'' policy; in Section~\ref{sec:pos-fast-growth}, we prove that these properties are useful in combination with many natural merge sorts.
Then, in Section~\ref{sec:TS-update}, we focus on \TS's actual update policy for the parameter~$\bt$, thereby identifying two weaknesses of this policy: using it may cause a non-negligible overhead, and may yet be rather inefficient.
Finally, in Section~\ref{sec:precise-bounds-PoS}, we mostly focus on the third policy presented above, which consists in choosing~$\bt$ as a function of the lengths of those runs we want to merge;
in particular, we shall prove that this policy enjoys the positive aspects of both previous update policies, while avoiding their drawbacks.

\section{Fast-growth and (tight) middle-growth properties}
\label{sec:fast-growth}

In this section, we focus on two novel properties of stable natural merge sorts, which we call \emph{fast-growth} and \emph{middle-growth}, and on a variant of the latter property, which we call \emph{tight middle-growth}.
These properties capture all \TS-like natural merge sorts invented in the last decade, and explain why these sorting algorithms require only~$\O(n + n \H)$ element moves and~$\O(n + n \min\{\H,\H^\ast\})$ element comparisons.
We will prove in Section~\ref{sec:pos-fast-growth} that many algorithms have these properties.

When applying a stable natural merge sort on an array~$A$, the elements of~$A$ are clustered into monotonic sub-arrays called \emph{runs}, and the algorithm consists in repeatedly merging consecutive runs into one larger run until the array itself contains only one run.
Consequently, each element may undergo several successive merge operations.
\emph{Merge trees}~\cite{BaNa13,Ju20,munro2018nearly} are a convenient way to represent the succession of runs that ever occur while~$A$ is being sorted.

\begin{definition}\label{def:merge-tree}
The \emph{merge tree} induced by a stable natural merge sort algorithm on an array~$A$ is the binary rooted tree~$\T$ defined as follows.
The nodes of~$\T$ are all the runs that were present in the initial array~$A$ or that resulted from merging two runs.
The runs of the initial array are the leaves of~$\T$, and when two consecutive runs~$R_1$ and~$R_2$ are merged with each other into a new run~$\ovR$, the run~$R_1$ spanning positions immediately to the left of those of~$R_2$, they form the left and the right children of the node~$\ovR$, respectively.
\end{definition}

Such trees ease the task of referring to several runs that might not have occurred simultaneously.
In particular, we will often refer to the~$i$\textsuperscript{th} \emph{ancestor} or a run~$R$, which is just~$R$ itself if~$i = 0$, or the parent, in the tree~$\T$, of the~$(i-1)$\textsuperscript{th} ancestor of~$R$ if~$i \geqslant 1$.
That ancestor will be denoted by~$\aR{i}$.

Before further manipulating these runs, let us first present some notation about runs and  their lengths, which we will frequently use.
We will commonly denote runs with capital letters, possibly with some index or adornment, and we will then denote the length of such a run with the same small-case letter and the same index or adornment.
For instance, runs named~$R$,~$R_i$,~$\aR{j}$,~$Q'$ and~$\overline{S}$ will have respective lengths~$r$,~$r_i$,~$\ar{j}$,~$q'$ and~$\overline{s}$.
Finally, we say that a run~$R$ is a \emph{left} run if it is the left child of its parent, and that it is a \emph{right} run if it is a right child.
The root of a merge tree is neither left nor right.

\begin{definition}\label{def:fast-growth}
We say that a stable natural merge sort algorithm~$\A$ has the \emph{fast-growth property} if it satisfies the following statement:
\begin{quote}
There exist an integer~$\ell \geqslant 1$ and a real number~$\alpha > 1$ such that, for every merge tree~$\T$ induced by~$\A$ and every run~$R$ at depth~$\ell$ or more in~$\T$, we have~$\ar{\ell} \geqslant \alpha r$.
\end{quote}
We also say that~$\A$ has the \emph{middle-growth property} if it satisfies the following statement:
\begin{quote}
There exists a real number~$\beta > 1$ such that, for every merge tree~$\T$ induced by~$\A$, every integer~$h \geqslant 0$ and every run~$R$ of height~$h$ in~$\T$, we have~$r \geqslant \beta^h$.
\end{quote}
Finally, we say that~$\A$ has the \emph{tight middle-growth property} if it satisfies the following statement:
\begin{quote}
There exists an integer~$\gamma \geqslant 0$ such that, for every merge tree~$\T$ induced by~$\A$, every integer~$h \geqslant 0$ and every run~$R$ of height~$h$ in~$\T$, we have~$r \geqslant 2^{h-\gamma}$.
\end{quote}
\end{definition}

Since every node of height~$h \geqslant 1$ in a merge tree is a run of length at least~$2$, each algorithm with the fast-growth property or with the tight middle-growth property also has the middle-growth property:
indeed, it suffices to choose~$\beta = \min\{2,\alpha\}^{1/\ell}$ in the first case, and~$\beta = 2^{1/(\gamma+1)}$ in the second one.
As a result, the first and third properties are stronger than the second one, and indeed they have stronger consequences.

\begin{theorem}\label{thm:fast-few-naïve}
Let~$\A$ be a stable natural merge sort algorithm with the fast-growth property.
If~$\A$ uses either the galloping or the naïve routine for merging runs, it requires~$\O(n + n\H)$ element comparisons and moves to sort arrays of length~$n$ and run-length entropy~$\H$.
\end{theorem}

\begin{proof}
Let~$\ell \geqslant 1$ and~$\alpha > 1$ be the integer and the real number mentioned in the definition of the statement ``$\A$ has the fast-growth property''.
Let~$A$ be an array of length~$n$ with~$\rho$ runs of lengths~$r_1,r_2,\ldots,r_\rho$, let~$\T$ be the merge tree induced by~$\A$ on~$A$, and let~$d_i$ be the depth of the run~$R_i$ in the tree~$\T$.

The algorithm~$\A$ uses~$\O(n)$ element comparisons and element moves to delimit the runs it will then merge and to make them non-decreasing.
Then, both the galloping and the naïve merging routine require~$\O(a+b)$ element comparisons and moves to merge two runs~$A$ and~$B$ of lengths~$a$ and~$b$.
Therefore, it suffices to prove that~$\sum_{R \in \T} r = \O(n + n \H)$.

Consider some leaf~$R_i$ of the tree~$\T$, and let~$k = \lfloor d_i / \ell \rfloor$.
The run~$\aR{k\ell}_i$ has length~$\ar{k\ell}_i \geqslant \alpha^k r_i$, and thus~$n \geqslant \ar{k\ell}_i \geqslant \alpha^k r_i$.
Hence,~$d_i+1 \leqslant \ell(k+1) \leqslant \ell \left(\log_\alpha(n/r_i) + 1\right)$, and we conclude that
\[
\sum_{R \in \T} r = \sum_{i=1}^\rho (d_i+1) r_i \leqslant
\ell \sum_{i=1}^\rho \left(r_i \log_\alpha(n/r_i) + r_i\right)
= \ell (n \H / \log_2(\alpha) + n) = \O(n + n\H).\qedhere\]
\end{proof}

Similar, weaker results also hold for algorithms with the (tight)  middle-growth property.

\begin{theorem}
\label{thm:middle-few-naïve}
Let~$\A$ be a stable natural merge sort algorithm with the middle-growth property.
If~$\A$ uses either the galloping or the naïve routine for merging runs, it requires~$\O(n \log(n))$ element comparisons and moves to sort arrays of length~$n$.
If, furthermore,~$\A$ has the tight middle-growth property, it requires at most~$n \log_2(n) + \O(n)$ element comparisons and moves to sort arrays of length~$n$.
\end{theorem}

\begin{proof}
Let us borrow the notations from the previous proof, and let~$\beta > 1$ be the real number mentioned in the definition of the statement ``$\A$ has the middle-growth property''.
Like in the proof of Theorem~\ref{thm:fast-few-naïve}, it suffices to show that~$\sum_{R \in \T} r = \O(n \log(n))$.
The~$d_i$\textsuperscript{th} ancestor of a run~$R_i$ is the root of~$\T$, and thus~$n \geqslant \beta^{d_i}$.
Hence,~$d_i \leqslant \log_\beta(n)$, which proves that
\[\sum_{R \in \T} r = \sum_{i=1}^\rho (d_i+1) r_i \leqslant
\sum_{i=1}^\rho (\log_\beta(n)+1) r_i =
(\log_\beta(n)+1) n = \O(n \log(n)).\]

Similarly, if~$\A$ has the tight middle-growth property, let~$\gamma$ be the integer mentioned in the definition of the statement ``$\A$ has the tight middle-growth property''.
This time,~$n \geqslant 2^{d_i-\gamma}$, which proves that~$d_i \leqslant \log_2(n) + \gamma$ and that
\[\sum_{R \in \T} r = \sum_{i=1}^\rho (d_i+1) r_i \leqslant
\sum_{i=1}^\rho (\log_2(n)+\gamma+1) r_i =
(\log_2(n)+\gamma+1) n = n \log_2(n) + \O(n). \qedhere\]
\end{proof}

Theorems~\ref{thm:fast-few-naïve} and~\ref{thm:middle-few-naïve} provide us with a simple framework for recovering well-known results on the complexity of many algorithms.
By contrast, Theorem~\ref{thm:middle-few} consists in new complexity guarantees on the number of element comparisons performed by algorithms with the middle-growth property, provided that they use the galloping routine.

\begin{theorem}\label{thm:middle-few}
Let~$\A$ be a stable natural merge sort algorithm with the middle-growth property.
If~$\A$ uses the galloping routine for merging runs, it requires~$\O(n + n\H^\ast)$ element comparisons to sort arrays of length~$n$ and dual run-length entropy~$\H^\ast$.
\end{theorem}

\begin{proof}
Let~$\beta > 1$ be the real number mentioned in the definition of the statement ``$\A$ has the middle-growth property''.
Let~$A$ be an array of length~$n$ with~$\sigma$ dual runs~$S_1,S_2,\ldots,S_\sigma$, and let~$\T$ be the merge tree induced by~$\A$ on~$A$.

The algorithm~$\A$ uses~$\O(n)$ element comparisons to delimit the runs it will then merge and to make them non-decreasing.
We prove now that merging these runs requires only~$\O(n + n \H^\ast)$ comparisons.
For every run~$R$ in~$\T$ and every integer~$i \leqslant \sigma$, let~$r_{\rightarrow i}$ be the number of elements of~$R$ that belong to the dual run~$S_i$.
Proposition~\ref{pro:2} states that merging two runs~$R$ and~$R'$ requires at most~$1 + \sum_{i=1}^\sigma \cost_{\bt}^\ast(r_{\rightarrow i}) + \cost_{\bt}^\ast(r'_{\rightarrow i})$ element comparisons.
Since less than~$n$ such merge operations are performed, and since~$n = \sum_{i=1}^\sigma s_i$ and~$n \H^\ast = \sum_{i=1}^\sigma s_i \log(n / s_i)$, it remains to show that
\[\sum_{R \in \T} \cost_{\bt}^\ast(r_{\rightarrow i}) = \O(s_i + s_i \log(n / s_i))\]
for all~$i \leqslant \sigma$.
Then, since~$\cost_{\bt}^\ast(m) \leqslant (\bt+1) \cost_0^\ast(m)$ for all parameter values~$\bt \geqslant 0$ and all~$m \geqslant 0$, we assume without loss of generality that~$\bt = 0$.

Consider now some integer~$h \geqslant 0$, and let~$\C_0(h) = \sum_{R \in \R_h} \cost_0^\ast(r_{\rightarrow i})$, where~$\R_h$ denotes the set of runs at height~$h$ in~$\T$.
Since no run in~$\R_h$ descends from another one, we already have~$\C_0(h) \leqslant 2 \sum_{R \in \R_h} r_{\rightarrow i} \leqslant 2 s_i$ and~$\sum_{R \in \R_h} r \leqslant n$.
Moreover, by definition of~$\beta$, each run~$R \in \R_h$ is of length~$r \geqslant \beta^h$, and thus~$|\R_h| \leqslant n / \beta^h$.

Let also~$f \colon x \mapsto 2 + 2 \log_2(x+1)$,~$g \colon x \mapsto x \, f(s_i / x)$ and~$\lambda = \lceil\log_\beta(n / s_i) \rceil$.
Both~$f$ and~$g$ are positive and concave on the interval~$(0,+\infty)$, thereby also being increasing.
It follows that, for all~$h \geqslant 0$,
\begin{align*}
\C_0(\lambda + h)
& \leqslant \sum_{R \in \R_{\lambda + h}} f(r_{\rightarrow i})
\leqslant |\R_{\lambda + h}| \, f\big({\textstyle\sum_{R \in \R_{\lambda + h}} r_{\rightarrow_i} / |\R_{\lambda + h}|}\big)
\leqslant g\big(|\R_{\lambda + h}|\big) \\
& \leqslant g\big(n / \beta^{\lambda + h}\big)
\leqslant g\big(s_i \beta^{-h}\big)
= \big(2+2 \log_2(\beta^h+1)\big) s_i \beta^{-h} \\
& \leqslant \big(2 + 2 \log_2(2 \beta^h)\big) s_i \beta^{-h}
= \big(4 + 2 h \log_2(\beta)\big) s_i \beta^{-h}.
\end{align*}
Inequalities on the first line hold by definition of~$\cost_0^\ast$, because~$f$ is concave, and because~$f$ is increasing;
inequalities on the second line hold because~$g$ is increasing and because~$|\R_h| \leqslant n / \beta^h$.

We conclude that
\begin{align*}
\sum_{R \in \T} \cost_0^\ast(r_{\rightarrow i})
& = \sum_{h \geqslant 0} \C_0(h)
= \sum_{h=0}^{\lambda-1} \C_0(h) + \sum_{h \geqslant 0}
\C_0(\lambda + h) \\
& \leqslant 2 \lambda s_i +
4 s_i \sum_{h \geqslant 0} \beta^{-h} +
2 \log_2(\beta) s_i \sum_{h \geqslant 0} h \beta^{-h} \\
& \leqslant
\O\big(s_i (1 + \log(n / s_i)\big) +
\O\big(s_i\big) +
\O\big(s_i\big) = \O\big(s_i+ s_i \log(n / s_i)\big). &&\qedhere
\end{align*}
\end{proof}

\section{Algorithms with the fast- and (tight) middle-growth properties}
\label{sec:pos-fast-growth}

In this section, we briefly present the algorithms mentioned in Section~\ref{sec:intro} and prove that each of them enjoys the fast-growth property and/or the (tight) middle-growth property.
Before treating these algorithms one by one, we first sum up our results.

\begin{theorem}\label{thm:fast}
The algorithms \TS, \aMS, \PoS, \PeS and \cASS have the fast-growth property.
\end{theorem}

An immediate consequence of Theorems~\ref{thm:fast-few-naïve} and~\ref{thm:middle-few} is that these algorithms sort arrays of length~$n$ and run-length entropy~$\H$ in time~$\O(n + n\H)$, which was already well-known; and that, if used with the galloping merging routine, they only need~$\O(n + n \H^\ast)$ comparisons to sort arrays of length~$n$ and dual run-length entropy~$\H^\ast$, which is a new result.

\begin{theorem}\label{thm:middle}
The algorithms \NMS, \ShS and \aSS have the middle-growth property.
\end{theorem}

Theorem~\ref{thm:middle-few} proves that, if these three algorithms are used with the galloping merging routine, they only need~$\O(n + n \H^\ast)$ comparisons to sort arrays of length~$n$ and dual run-length entropy~$\H^\ast$.
By contrast, observe that they can be implemented by using a stack, following \TS's own implementation, but where only the two top runs of the stack could be merged.
It is proved in~\cite{Ju20} that such algorithms may require~$ \omega(n + n \H)$ comparisons to sort arrays of length~$n$ and run-length entropy~$\H$.
Hence, Theorem~\ref{thm:fast-few-naïve} shows that these three algorithms do \emph{not} have the fast-growth property.

\begin{theorem}\label{thm:tight-middle}
The algorithms \NMS, \ShS and \PoS have the tight middle-growth property.
\end{theorem}

Theorem~\ref{thm:middle-few-naïve} proves that these algorithms sort arrays of length~$n$ in time~$n \log_2(n) + \O(n)$, which was already well-known.
In Section~\ref{sec:precise-bounds-PoS}, we will further improve our upper bounds on the number of comparisons that these algorithms require when sorting arrays of length~$n$ and dual run-length entropy~$\H^\ast$.

Note that the algorithms \aSS, \aMS and \TS may require more than~$n \log_2(n) + \O(n)$ comparisons to sort arrays of length~$n$, which proves that they do not have the tight middle-growth property.
In Section~\ref{sec:precise-bounds-PoS}, we will prove that \cASS and \PeS also fail to have the this property, although they enjoy still complexity upper bounds similar to those of algorithms with the tight middle-growth property.

\subsection{Algorithms with the fast-growth property}
\label{subsec:other-fast}

\subsubsection{\PoS}
\label{subsubsec:PoS}

The algorithm \PoS is best defined by introducing the notion of \emph{power} of a run endpoint or of a run, and then characterising the merge trees that \PoS induces.

\begin{definition}\label{def:PoS:power}
Let~$A$ be an array of length~$n$, whose run decomposition consists of runs~$R_1,R_2,\ldots,R_\rho$, ordered from left to right.
For all integers~$i \leqslant \rho$, let~$e_i = r_1+\ldots+r_i$.
We also abusively set~$e_{-1} = -\infty$ and~$e_{\rho+1} = n$.

When~$0 \leqslant i \leqslant \rho$, let~$\intt(i)$ denote the half-open interval~$(e_{i-1}+e_i,e_i+e_{i+1}]$.
The \emph{power} of~$e_i$, denoted by~$p_i$, is then defined as the least integer~$p$ such that~$\intt(i)$ contains an element of the set~$\{k n / 2^{p-1} \colon k \in \Z\}$.
Thus, we (abusively) have~$p_0 = -\infty$ and~$p_\rho = 0$, because~$\intt(0) = (-\infty,r_1/2]$ and~$\intt(\rho) = (2n-r_\rho,2n]$.

Finally, let~$R_{i \ldots j}$ be a run obtained by merging consecutive runs~$R_i,R_{i+1},\ldots,R_j$.
The \emph{power} of the run~$R_{i \ldots j}$ is defined as~$\max\{p_{i-1},p_j\}$.
\end{definition}

\begin{lemma}\label{lem:unique-min-power}
Each non-empty sub-interval~$I$ of the set~$\{0,\ldots,\rho\}$ contains exactly one integer~$i$ such that~$p_i \leqslant p_j$ for all~$j \in I$.
\end{lemma}

\begin{proof}
Assume that the integer~$i$ is not unique.
Since~$e_0$ is the only endpoint with power~$-\infty$, we know that~$0 \notin I$.
Then, let~$a$ and~$b$ be elements of~$I$ such that~$a < b$ and~$p_a = p_b \leqslant p_j$ for all~$j \in I$, and let~$p = p_a = p_b$.
By definition of~$p_a$ and~$p_b$, there exist odd integers~$k$ and~$\ell$ such that~$k n / 2^{p - 1} \in \intt(a)$ and~$\ell n / 2^{p - 1} \in \intt(b)$.
Since~$\ell \geqslant k+1$, the fraction~$(k + 1) n / 2^{p-1}$ belongs to some interval~$\intt(j)$ such that~$a \leqslant j \leqslant b$.
But since~$k+1$ is even, we know that~$p_j < p$, which is absurd.
This invalidates our initial assumption and completes the proof.
\end{proof}

\begin{corollary}
\label{cor:construction}
Let~$R_1,\ldots,R_\rho$ be the run decomposition of an array~$A$. There is exactly one tree~$\T$ that is induced on~$A$ and
in which every inner node has a smaller power than its children.
Furthermore, for every run~$R_{i \ldots j}$ in~$\T$, we have~$\max\{p_{i-1},p_j\} < \min\{p_i,p_{i+1},\ldots,p_{j-1}\}$.
\end{corollary}

\begin{proof}
Given a merge tree~$\T$, let us prove that the following statements are equivalent:
\begin{itemize}
\item[$\mathsf{S}_1$:]\label{item:unique:1}
each inner node of~$\T$ has a smaller power than its children;
\item[$\mathsf{S}_2$:]\label{item:unique:2}
each run~$R_{i \ldots j}$ that belongs to~$\T$ has a power that is smaller than all of~$p_i,\ldots,p_{j-1}$;
\item[$\mathsf{S}_3$:]\label{item:unique:3}
if a run~$R_{i \ldots j}$ is an inner node of~$\T$, its children are the two runs~$R_{i \ldots k}$ and~$R_{k+1 \ldots j}$ such that~$p_k = \min\{p_i,\ldots,p_{j-1}\}$.
\end{itemize}

First, if~$\mathsf{S}_1$ holds, we prove~$\mathsf{S}_3$ by induction on the height~$h$ of the run~$R_{i \ldots j}$.
Indeed, if the restriction of~$\mathsf{S}_3$ to runs of height less than~$h$ holds, let~$R_{i \ldots k}$ and~$R_{k+1 \ldots j}$ be the children of a run~$R_{i \ldots j}$ of height~$h$.
If~$i < k$, the run~$R_{i \ldots k}$ has two children~$R_{i \ldots \ell}$ and~$R_{\ell+1 \ldots k}$ such that~$p_\ell = \min\{p_i,\ldots,p_{k-1}\}$, and the powers of these runs, i.e.,~$\max\{p_{i-1},p_\ell\}$ and~$\max\{p_\ell,p_k\}$, are greater than the power of~$R_{i \ldots k}$, i.e.,~$\max\{p_{i-1},p_k\}$, which proves that~$p_\ell > p_k$.
It follows that~$p_k = \min\{p_i,\ldots,p_k\}$, and one proves similarly that~$p_k = \min\{p_k,\ldots,p_{j-1}\}$, thereby showing that~$\mathsf{S}_3$ also holds for runs of height~$h$.

Then, if~$\mathsf{S}_3$ holds, we prove~$\mathsf{S}_2$ by induction on the depth~$d$ of the run~$R_{i \ldots j}$.
Indeed, if the restriction of~$\mathsf{S}_2$ to runs of depth less than~$d$ holds, let~$R_{i \ldots k}$ and~$R_{k+1 \ldots j}$ be the children of a run~$R_{i \ldots j}$ of depth~$d$.
Lemma~\ref{lem:unique-min-power} and~$\mathsf{S}_3$ prove that~$p_k$ is the unique smallest element of~$\{p_i,\ldots,p_{j-1}\}$, and the induction hypothesis proves that~$\max\{p_{i-1},p_j\} < p_k$.
It follows that both powers~$\max\{p_{i-1},p_k\}$ and~$\max\{p_k,p_j\}$ are smaller than all of~$p_i,\ldots,p_{k-1},p_{k+1},\ldots,p_{j-1}$, thereby showing that~$\mathsf{S}_2$ also holds for runs of depth~$d$.

Finally, if~$\mathsf{S}_2$ holds, let~$R_{i \ldots j}$ be an inner node of~$\T$, with children~$R_{i \ldots k}$ and~$R_{k+1 \ldots j}$.
Property~$\mathsf{S}_2$ ensures that~$\max\{p_{i-1},p_j\} < p_k$, and thus that~$\max\{p_{i-1},p_j\}$ is smaller than both~$\max\{p_{i-1},p_k\}$ and~$\max\{p_k,p_j\}$, i.e., that~$R_{i \ldots j}$ has a smaller power that its children, thereby proving~$\mathsf{S}_1$.

In particular, once the array~$A$ and its run decomposition~$R_1,\ldots,R_\rho$ are fixed,~$\mathsf{S}_3$ provides us with a deterministic top-down construction of the unique merge tree~$\T$ induced on~$A$ and that satisfies~$\mathsf{S}_1$:
the root of~$\T$ must be the run~$R_{1 \ldots \rho}$ and, provided that some run~$R_{i \ldots j}$ belongs to~$\T$, where~$i < j$, Lemma~\ref{lem:unique-min-power} proves that the integer~$k$ mentioned in~$\mathsf{S}_3$ is unique, which means that~$\mathsf{S}_3$ unambiguously describes the children of~$R_{i \ldots j}$ in the tree~$\T$.

This proves the first claim of Corollary~\ref{cor:construction}, and the second claim of Corollary~\ref{cor:construction} follows from the equivalence between the statements~$\mathsf{S}_1$ and~$\mathsf{S}_2$.
\end{proof}

This leads to the following characterisation of the algorithm \PoS, which is proved in~\cite[Lemma 4]{munro2018nearly}, and which Corollary~\ref{cor:construction} allows us to consider as an alternative definition of \PoS.

\begin{definition}\label{def:PoS}
In every merge tree that \PoS induces, inner nodes have a smaller power than their children.
\end{definition}

\begin{lemma}\label{lem:PoS:power}
Let~$\T$ be a merge tree induced by \PoS, let~$R$ be a run of~$\T$ with power~$p$, and let~$\aR{2}$ be its grandparent.
We have~$2^{p-2} r < n < 2^p \ar{2}$.
\end{lemma}

\begin{proof}
Let~$R_{i \ldots j}$ be the run~$R$.
Without loss of generality, we assume that~$p = p_j$, the case~$p = p_{i-1}$ being entirely similar.
Corollary~\ref{cor:construction} states that all of~$p_i,\ldots,p_{j-1}$ are larger than~$p$, and therefore that~$p \leqslant \min\{p_i,\ldots,p_j\}$.
Thus, the union of intervals~$\intt(i) \cup \ldots \cup \intt(j) = (e_{i-1}+e_i,e_j+e_{j+1}]$ does not contain any element of the set~$\S = \{k n / 2^{p-2} \colon k \in \Z\}$.
Consequently, the bounds~$e_{i-1}+e_i$ and~$e_j+e_{j+1}$ are contained between two consecutive elements of~$\S$, i.e., there exists an integer~$\ell$ such that
\[\ell n / 2^{p-2} \leqslant e_{i-1}+e_i \leqslant
e_j+e_{j+1} < (\ell+1) n / 2^{p-2}.\]
We conclude that
\[r = e_j - e_{i-1} \leqslant (e_j+e_{j+1}) - (e_{i-1}+e_i) < n / 2^{p-2}.\]

We prove now that~$n \leqslant 2^p \ar{2}$.
To that end, we assume that both~$R$ and~$\aR{1}$ are left children, the other possible cases being entirely similar.
There exist integers~$u$ and~$v$ such that~$\aR{1} = R_{i \ldots u}$ and~$\aR{2} = R_{i \ldots v}$.
Hence,~$\max\{p_{i-1},p_u\} < \max\{p_{i-1},p_j\} = p$, which shows that~$p_u < p_j = p$.
Thus, both intervals~$\intt(j)$ and~$\intt(u)$, which are subintervals of~$(2e_{i-1},2e_v]$, contain elements of the set~$\S' = \{k n / 2^{p-1} \colon k \in \Z\}$.
This means that there exist two integers~$k$ and~$\ell$ such that~$2e_{i-1} < k n / 2^{p-1} < \ell n / 2^{p-1} \leqslant 2 e_v$, from which we conclude that
\[\ar{2} = e_v - e_{i-1} > (\ell-k) n / 2^p \geqslant n / 2^p. \qedhere\]
\end{proof}

\begin{theorem}\label{thm:fast-growth-PoS}
The algorithm \PoS has the fast-growth property.
\end{theorem}

\begin{proof}
Let~$\T$ be a merge tree induced by \PoS.
Then, let~$R$ be a run in~$\T$, and let~$p$ and~$\ap{3}$ be the respective powers of the runs~$R$ and~$\aR{3}$.
Definition~\ref{def:PoS} ensures that~$p \geqslant \ap{3}+3$, and therefore Lemma~\ref{lem:PoS:power} proves that
\[2^{\ap{3}+1} r \leqslant 2^{p-2} r < n < 2^{\ap{3}} \ar{5}.\]
This means that~$\ar{5} \geqslant 2 r$, and therefore that \PoS has the fast-growth property.
\end{proof}

\subsubsection{\PeS}
\label{subsubsec:PeS}

Like its sibling \PoS, the algorithm \PeS is best defined by characterizing the merge trees it induces.

\begin{definition}\label{def:tree:PeS}
Let~$\T$ be the merge tree induced by \PeS on an array~$A$.
The children of each internal node~$R_{i \ldots j}$ of~$\T$ are the runs~$R_{i \ldots k}$ and~$R_{k+1 \ldots j}$ for which the quantity
\[ |2 e_k - e_j - e_{i-1}|\]
is minimal.
In case of equality, the integer~$k$ is chosen to be as small as possible.
\end{definition}

\begin{proposition}\label{pro:fast-growth-PeS}
The algorithm \PeS has the fast-growth property.
\end{proposition}

\begin{proof}
Let~$\T$ be a merge tree induced by \PeS, and let~$R$ be a run in~$\T$.
We prove that~$\ar{3} \geqslant 2 r$.
Indeed, let us assume that a majority of the runs~$R = \aR{0}$,~$\aR{1}$ and~$\aR{2}$ are left runs.
The situation is entirely similar if a majority of these runs are right runs.

Let~$i < j$ be the two smallest integers such that~$\aR{i}$ and~$\aR{j}$ are left runs, and let~$S$ and~$T$ be their right siblings, as illustrated in Figure~\ref{fig:pro:fast-growth-PeS}.
We can write these runs as~$\aR{i} = R_{w+1 \ldots x}$, {}$S = R_{x+1 \ldots y}$,~$\aR{j} = R_{v+1 \ldots y}$ and~$T = R_{y+1 \ldots z}$ for some integers~$v \leqslant w < x < y < z$.

Definition~\ref{def:tree:PeS} states that
\[|\ar{j} - t| = |2 e_y - e_z - e_v| \leqslant |2 e_{y-1} - e_z - e_v| = |\ar{j} - t - 2 r_y|.\]
Thus,~$\ar{j} - t - 2 r_y$ is negative, i.e.,~$t > \ar{j} - 2 r_y$, and
\[\ar{3} \geqslant \ar{j+1} = \ar{j} + t >
2 \ar{j} - 2 r_y = 2 (e_{y-1} - e_v) \geqslant 2 (e_x - e_w) = 2 \ar{i} \geqslant 2 r. \qedhere\]
\end{proof}

\begin{figure}[t]
\begin{center}
\begin{tikzpicture}[scale=0.65]
\foreach \i/\j in {-3/2.8,0/1,3/1,6/2,9/2.8}{
 \draw[very thick,draw=black!20,fill=black!20]
 (\i,\j) -- (\i,-0.2) -- 
 (\i-0.2,-0.2) -- (\i-0.2,\j) -- cycle;
 \draw[very thick,draw=white,fill=white]
 (\i,\j) -- (\i,-0.2) -- 
 (\i+0.2,-0.2) -- (\i+0.2,\j) -- cycle;
 \draw[very thick,draw=white,fill=white]
 (\i-0.4,\j) -- (\i-0.4,-0.2) -- 
 (\i-0.2,-0.2) -- (\i-0.2,\j) -- cycle;
}
\foreach \i/\j/\k in {0/0/3,3/0/3,-3/1/9,6/1/3,-3/2/12}{
 \draw[very thick] (\i,\j) -- (\i+\k-0.2,\j) --
 (\i+\k-0.2,\j+0.8) -- (\i,\j+0.8) -- cycle;
}

\node[anchor=south] at (1.4,0) {$\aR{i}$};
\node[anchor=south] at (4.4,0) {$S$};
\node[anchor=south] at (1.4,1) {$\aR{j}$};
\node[anchor=south] at (7.4,1) {$T$};
\node[anchor=south] at (2.9,2) {$\aR{j+1}$};
\node[anchor=north] at (-3,-0.15) {$e_v$};
\node[anchor=north] at (0,-0.15) {$e_w$};
\node[anchor=north] at (3,-0.15) {$e_x$};
\node[anchor=north] at (6,-0.15) {$e_y$};
\node[anchor=north] at (9,-0.15) {$e_z$};
\end{tikzpicture}
\end{center}
\vspace{-5mm}
\caption{Runs~$\aR{i}$,~$\aR{j}$,~$\aR{j+1}$, their siblings~$S$ and~$T$, and endpoints~$e_v \leqslant e_w < e_x < e_y < e_z$.\label{fig:pro:fast-growth-PeS}}
\end{figure}

\subsubsection{\CASS}
\label{subsubsec:cASS}

The algorithm \cASS is presented in Algorithm~\ref{alg:cASS}.
It is based on an \emph{ad hoc} tool, which we call \emph{level} of a run :
the level of a run~$R$ of length~$r$ is defined as the number~$\ell = \lfloor \log_2(r) \rfloor$.
In practice, and following our naming conventions, let~$\al{i}$ and~$\ell_i$ denote the respective levels of the runs~$\aR{i}$ and~$R_i$.

\begin{algorithm}[h]
\begin{small}
\SetArgSty{texttt}
\DontPrintSemicolon
\Input{Array~$A$ to sort}
\Result{The array~$A$ is sorted into a single run.
That run remains on the 
stack.}
\Note{The height of the stack~$\S$ is denoted by~$h$; its~$i$\textsuperscript{th} deepest run by~$R_i$; the length of~$R_i$~by~$r_i$.
Finally, we set~$\ell_i = \lfloor \log_2(r_i) \rfloor$.
When two consecutive runs of~$\S$ are merged, they are replaced, in~$\S$, with the run resulting from the merge. \justifying}
\BlankLine {}$\S \gets~$ an empty stack

\While(\label{main-loop:cASS:start}){\true}{
 \If(\label{test:cASS:case-1})
 {\textrm{$h \geqslant 3$ and~$\ell_{h-2}
 \leqslant \max\{\ell_{h-1},\ell_h\}$}}
 {merge the runs~$R_{h-2}$ and~$R_{h-1}$\label{alg:cASS:merge}}
 \ElseIf{\textrm{the end of the array has not yet been reached}}
 {find a new monotonic run~$R$, make it non-decreasing,
 and push it onto~$\S$\label{cASS:case-push}}
 \Else{break\label{algline:cASS:end_inner_loop}}
}
\While{$h \geqslant 2$\label{alg:cASS:trigger:collapse}}{
 merge the runs~$R_{h-1}$ and~$R_h$
 \label{alg:cASS:collapse}
}
\end{small}
\caption{\cASS\label{alg:cASS}}
\end{algorithm}

Observe that appending a fictitious run of length~$2n$ to the array~$A$ and stopping our sequence of merges just before merging that fictitious run does not modify the sequence of merges performed by the algorithm, but allows us to assume that every merge was performed in line~\ref{alg:cASS:merge}.
Therefore, we work below under that assumption.

Our proof is based on the following result, which was stated and proved in~\cite[Lemma 7]{Ju20}.

\begin{lemma}\label{lemma:cASS:invariant}
Let~$\S = (R_1,R_2,\ldots,R_h)$ be a stack obtained while executing \cASS.
We have (i)~$\ell_i \geqslant \ell_{i+1}+1$ whenever~$1 \leqslant i \leqslant h-3$ and (ii)~$\ell_{h-3} \geqslant \ell_{h-1}+1$ if~$h \geqslant 4$.
\end{lemma}

Then, Lemmas~\ref{lem:cASS::right} and~\ref{lem:cASS::left} focus on inequalities involving the levels of a given run~$R$ belonging to a merge tree induced by \cASS, and of its ancestors.
In each case, the stack just before the run~$R$ is merged is denoted by~$\S = (R_1,R_2,\ldots,R_h)$

\begin{lemma}\label{lem:cASS::right}
If~$\aR{1}$ is a right run,~$\al{2} \geqslant \ell+1$.
\end{lemma}

\begin{proof}
The run~$R$ coincides with either~$R_{h-2}$ or~$R_{h-1}$, and~$\aR{1}$ is the parent of the runs~$R_{h-2}$ and~$R_{h-1}$.
Hence, the run~$R_{h-3}$ descends from the left sibling of~$\aR{1}$ and from~$\aR{2}$.
Thus, inequalities (i) and (ii) prove that~$\al{2} \geqslant \ell_{h-3} \geqslant \max\{\ell_{h-2},\ell_{h-1}\}+1 \geqslant \ell+1$.
\end{proof}

\begin{lemma}\label{lem:cASS::left}
If~$R$ is a left run,~$\al{2} \geqslant \ell+1$.
\end{lemma}

\begin{proof}
Since~$R$ is a left run, it coincides with~$R_{h-2}$, and~$\ell_{h-2} \leqslant \max\{\ell_{h-1},\ell_h\}$.
Then, if~$\aR{1}$ is a right run, Lemma~\ref{lem:cASS::right} already proves that~$\al{2} \geqslant \ell+1$.
Otherwise,~$\aR{1}$ is a left run, and the run~$R_h$ descends from~$\aR{2}$, thereby proving that
\[\ar{2} \geqslant r + \max\{r_{h-1},r_h\}
\geqslant {2^\ell + 2^{\max\{\ell_{h-1},\ell_h\}}}
\geqslant 2^{\ell+1},\]
and thus that~$\al{2} \geqslant \ell+1$ too.
\end{proof}

\begin{proposition}\label{pro:fast-growth-cASS}
The algorithm \cASS has the fast-growth property.
\end{proposition}

\begin{proof}
Let~$\T$ be a merge tree induced by \cASS, and let~$R$ be a run in~$\T$.
Lemma~\ref{lem:cASS::right} shows that~$\al{3} \geqslant \al{2} \geqslant \ell+1$ if~$\aR{1}$ is a right run, and Lemma~\ref{lem:cASS::left} shows that~$\al{3} \geqslant \al{1}+1 \geqslant \ell+1$ if~$\aR{1}$ is a left run.
Thus, the inequality~$\al{3} \geqslant \ell+1$ is valid in both cases, and we show similarly that~$\al{6} \geqslant \al{3}+1$.
It follows that~$\ar{6} \geqslant 2^{\al{6}} \geqslant 2^{\ell+2} \geqslant 2r$.
\end{proof}

\subsubsection{\TS}
\label{subsubsec:TS}

The algorithm \TS is presented in Algorithm~\ref{alg:TS}.

\begin{algorithm}[h]
\begin{small}
\SetArgSty{texttt}
\DontPrintSemicolon
\Input{Array~$A$ to sort}
\Result{The array~$A$ is sorted into a single run.
That run remains on the stack.}
\Note{The height of the stack~$\S$ is denoted by~$h$; its~$i$\textsuperscript{th} deepest run by~$R_i$; the length of~$R_i$ by~$r_i$.
When two consecutive runs of~$\S$ are merged, they are replaced, in~$\S$, with the run resulting from the merge. \justifying}
\BlankLine
$\S \gets~$ an empty stack
\While(\label{main-loop:start}){\true}{
 \If(\label{test:case-1})
 {\textrm{$h \geqslant 3$ and~$r_{h-2} < r_h$}}
 {merge the runs~$R_{h-2}$ and~$R_{h-1}$
 \tcp*[f]{case \#1}\label{case-1}}
 \ElseIf(\label{test:case-2})
 {\textrm{$h \geqslant 2$ and~$r_{h-1} \leqslant r_h$}}
 {merge the runs~$R_{h-1}$ and~$R_h$
 \tcp*[f]{case \#2}\label{case-2}}
 \ElseIf(\label{test:case-3})
 {\textrm{$h \geqslant 3$ and~$r_{h-2} \leqslant r_{h-1} + r_h$}}
 {merge the runs~$R_{h-1}$ and~$R_h$
 \tcp*[f]{case \#3}\label{case-3}}
 \ElseIf(\label{test:case-4})
 {\textrm{$h \geqslant 4$ and~$r_{h-3} \leqslant r_{h-2} + r_{h-1}$}}
 {merge the runs~$R_{h-1}$ and~$R_h$
 \tcp*[f]{case \#4}\label{case-4}}
 \ElseIf{\textrm{the end of the array has not yet been reached}}
 {find a new monotonic run~$R$, make it non-decreasing, and push it onto~$\S$\label{case-push}}
 \Else{break\label{algline:end_inner_loop}}
}
\While{$h \geqslant 2$}{
 merge the runs~$R_{h-1}$ and~$R_h$
 \label{alg:collapse}
}
\end{small}
\caption{\TS\label{alg:TS}}
\end{algorithm}

We say that a run~$R$ is a \#1-, a \#2-, a \#3 or a \#4-run if is merged in line~\ref{case-1},~\ref{case-2},~\ref{case-3} or~\ref{case-4}, respectively.
Like in Section~\ref{subsubsec:cASS}, appending a fictitious run of length~$2n$ that we will avoid merging allows us to assume that every run is merged in line~\ref{case-1},~\ref{case-2},~\ref{case-3} or~\ref{case-4}, i.e., is a \#1-, a \#2-, a \#3 or a \#4-run.

Our proof is then based on the following result, which extends~\cite[Lemma 5]{auger2018worst} by adding the inequality~(v).

\begin{lemma}\label{lemma:invariant}
Let~$\S = (R_1,R_2,\ldots,R_h)$ be a stack obtained while executing \TS.
We have
\begin{enumerate}[label=(\roman*)]
\item {}$r_i > r_{i+1} + r_{i+2}$ whenever~$1 \leqslant i \leqslant h-4$,
\item {}$3 r_{h-2} > r_{h-1}$ if~$h \geqslant 3$,
\item {}$r_{h-3} > r_{h-2}$ if~$h \geqslant 4$,
\item {}$r_{h-3} + r_{h-2} > r_{h-1}$ if~$h \geqslant 4$, and
\item {}$\max\{r_{h-3}/2,4 r_h\} > r_{h-1}$ if~$h \geqslant 4$.
\end{enumerate}
\end{lemma}

\begin{proof}
Lemma 5 from~\cite{auger2018worst} already proves the inequalities (i) to (iv).
Therefore, we prove, by a direct induction on the number of (push or merge) operations performed before obtaining the stack~$\S$, that~$\S$ also satisfies (v).

When the algorithm starts, we have~$h \leqslant 3$, and therefore there is nothing to prove in that case.
Then, when a stack~$\S = (R_1,R_2,\ldots,R_h)$ obeying the inequalities (i) to (v) is transformed into a stack~$\overline{\S} = (\ovR_1,\ovR_2, \ldots,\ovR_{\ovh})$
\begin{itemize}
\item by inserting a run, cases \#2 and \#3 just failed to occur and~$\ovh = h+1$, so that
\[\ovr_{\ovh-3} = r_{h-2} > r_{h-1} + r_h > 2 r_h = 2 \ovr_{\ovh-1};\]
\item by merging the runs~$R_{h-1}$ and~$R_h$, we have~$\ovh = h-1$ and
\[\ovr_{\ovh-3} = r_{h-4} > r_{h-3} + r_{h-2} > 2 r_{h-2} = 2 \ovr_{\ovh-1};\]
\item by merging the runs~$R_{h-2}$ and~$R_{h-1}$, case \#1 just occurred and~$\ovh = h-1$, so that
\[4 \ovr_{\ovh} = 4 r_h > 4 r_{h-2}
> r_{h-2} + r_{h-1} = \ovr_{\ovh-1}.\]
\end{itemize}

In each case,~$\overline{\S}$ satisfies (v), which completes the induction and the proof.
\end{proof}

Then, Lemmas~\ref{lem::right} and~\ref{lem::left} focus on inequalities involving the lengths of a given run~$R$ belonging to a merge tree induced by \TS and of its ancestors.
In each case, the stack just before the run~$R$ is merged is denoted by~$\S = (R_1,R_2,\ldots,R_h)$.

\begin{lemma}\label{lem::right}
If~$\aR{1}$ is a right run,~$\ar{2} \geqslant 4r/3$.
\end{lemma}

\begin{proof}
Let~$S$ be the left sibling of the run~$\aR{1}$, and let~$i$ be the integer such that~$R = R_i$.
If~$i = h-2$, (iii) shows that~$r_{i-1} > r_i$.
If~$i = h-1$, (ii) shows that~$r_{i-1} > r_i / 3$.
In both cases, the run~$R_{i-1}$ descends from~$\aR{2}$, and thus~$\ar{2} \geqslant r + r_{i-1} \geqslant 4 r / 3$.

Finally, if~$i = h$, the run~$R$ is a \#2-, \#3- or \#4-right run, which means both that~$r_{h-2} \geqslant r$ and that~$R_{h-2}$ descends from~$S$.
It follows that~$\ar{2} \geqslant r + r_{h-2} \geqslant 2 r$.
\end{proof}

\begin{lemma}\label{lem::left}
If~$R$ is a left run,~$\ar{2} \geqslant 5 r / 4$.
\end{lemma}

\begin{proof}
We treat four cases independently, depending on whether~$R$ is a \#1-, a \#2-, a \#3- or a \#4-left run.
In each case, we assume that the run~$\aR{1}$ is a left run, since Lemma~\ref{lem::right} already proves that~$\ar{2} \geqslant 4r/3$ when~$\aR{1}$ is a right run.
\begin{itemize}
\item If~$R$ is a \#1-left run, the run~$R = R_{h-2}$ is merged with~$R_{h-1}$ and~$r_{h-2} < r_h$.
Since~$\aR{1}$ is a left run,~$R_h$ descends from~$\aR{2}$, and thus~$\ar{2} \geqslant r + r_h \geqslant 2r$.

\item If~$R$ is a \#2-left run, the run~$R = R_{h-1}$ is merged with~$R_h$ and~$r_{h-1} \leqslant r_h$.
It follows, in that case, that~$\ar{2} \geqslant \ar{1} = r + r_h \geqslant 2 r$.

\item If~$R$ is a \#3-left run, the run~$R = R_{h-1}$ is merged with~$R_h$, and~$r_{h-2} \leqslant r_{h-1} + r_h = \ar{1}$.
Due to this inequality, our next journey through the loop of lines~\ref{main-loop:start} to~\ref{algline:end_inner_loop} must trigger another merge.
Since~$\aR{1}$ is a left run, that merge must be a \#1-merge, which means that~$r_{h-3} < \ar{1}$.
Consequently, (v) proves that~$\ar{1} \geqslant \max\{r_{h-3},r_{h-1}+r_h\} \geqslant 5 r_{h-1}/4 = 5 r/4$.

\item We prove that~$R$ cannot be a \#4-left run.
Indeed, if~$R$ is a \#4-left run, the run~$R = R_{h-1}$ is merged with~$R_h$, and we both have~$r_{h-2} > \ar{1}$ and~$r_{h-3} \leqslant r_{h-2} + r_{h-1} \leqslant r_{h-2} + \ar{1}$.
Due to the latter inequality, our next journey through the loop of lines~\ref{main-loop:start} to~\ref{algline:end_inner_loop} must trigger another merge.
Since~$\ar{1} < r_{h-2} < r_{h-3}$, that new merge cannot be a \#1-merge, and thus~$\aR{1}$ is a right run, contradicting our initial assumption. \qedhere
\end{itemize}
\end{proof}

\begin{proposition}\label{pro:fast-growth-TS}
The algorithm \TS has the fast-growth property.
\end{proposition}

\begin{proof}
Let~$\T$ be a merge tree induced by \TS, and let~$R$ be a run in~$\T$.
Lemma~\ref{lem::right} shows that~$\ar{3} \geqslant \ar{2} \geqslant 4r/3$ if~$\aR{1}$ is a right run, and Lemma~\ref{lem::left} shows that~$\ar{3} \geqslant 5 \ar{1}/4 \geqslant 5r/4$ if~$\aR{1}$ is a left run.
Hence, in both cases,~$\ar{3} \geqslant 5 r/4$.
\end{proof}

\subsubsection{\aMS}
\label{subsubsec:aMS}

The algorithm \aMS is parametrised by a real number~$\alpha > 1$ and
is presented in Algorithm~\ref{alg:aMS}.

\begin{algorithm}[h]
\begin{small}
\SetArgSty{texttt}
\DontPrintSemicolon
\Input{Array~$A$ to sort, parameter~$\alpha > 1$}
\Result{The array~$A$ is sorted into a single run.
That run remains on the stack.}
\Note{The height of the stack~$\S$ is denoted by~$h$; its~$i$\textsuperscript{th} deepest run by~$R_i$; the length of~$R_i$~by~$r_i$.
When two consecutive runs of~$\S$ are merged, they are replaced, in~$\S$, with the run resulting from the merge. \justifying}
\BlankLine {}$\S \gets~$ an empty stack

\While(\label{aMS:main-loop:start}){\true}{
 \If(\label{test:aMS:case-1})
 {\textrm{$h \geqslant 3$ and~$r_{h-2} < r_h$}}
 {merge the runs~$R_{h-2}$ and~$R_{h-1}$
 \tcp*[f]{case \#1}\label{aMS:case-1}}
 \ElseIf(\label{test:aMS:case-2})
 {\textrm{$h \geqslant 2$ and~$r_{h-1} < \alpha r_h$}}
 {merge the runs~$R_{h-1}$ and~$R_h$
 \tcp*[f]{case \#2}\label{aMS:case-2}}
 \ElseIf(\label{test:aMS:case-3})
 {\textrm{$h \geqslant 3$ and~$r_{h-2} < \alpha r_{h-1}$}}
 {merge the runs~$R_{h-1}$ and~$R_h$
 \tcp*[f]{case \#3}\label{aMS:case-3}}
 \ElseIf{\textrm{the end of the array has not yet been reached}}
 {find a new monotonic run~$R$, make it non-decreasing,
 and push it onto~$\S$\label{aMS:case-push}}
 \Else{break\label{algline:aMS:end_inner_loop}}
}
\While{$h \geqslant 2$}{
 merge the runs~$R_{h-1}$ and~$R_h$
 \label{alg:aMS:collapse}
}
\end{small}
\caption{\aMS\label{alg:aMS}}
\end{algorithm}

Like in Section~\ref{subsubsec:TS}, we say that a run~$R$ is
a \#1-, a \#2- or a \#3-run if is merged in line~\ref{aMS:case-1},~\ref{aMS:case-2}~or~\ref{aMS:case-3}.
In addition, still like in Sections~\ref{subsubsec:cASS} and~\ref{subsubsec:TS}, we safely assume that each run is a \#1-, a \#2- or a \#3-run.

Our proof is then based on the following result, which extends~\cite[Theorem 14]{BuKno18} by adding the inequalities (ii) and (iii).

\begin{lemma}\label{lemma:aMS:invariant}
Let~$\S = (R_1,R_2,\ldots,R_h)$ be a stack obtained while executing \mbox{\aMS}.
We have
\begin{enumerate}[label=(\roman*)]
\item {}$r_i \geqslant \alpha r_{i+1}$ whenever~$1 \leqslant i \leqslant h-3$,
\item {}$r_{h-2} \geqslant (\alpha-1) r_{h-1}$ if~$h \geqslant 3$, and
\item {}$\max\{r_{h-2}/\alpha,\alpha r_h/(\alpha-1)\} \geqslant r_{h-1}$ if~$h \geqslant 3$.
\end{enumerate}
\end{lemma}

\begin{proof}
Theorem 14 from~\cite{BuKno18} already proves the inequality (i). Therefore, we prove, by a direct induction on the number of (push or merge) operations performed before obtaining the stack~$\S$, that~$\S$ satisfies (ii) and (iii).

When the algorithm starts, we have~$h \leqslant 2$, and therefore there is nothing to prove in that case.
Then, when a stack~$\S = (R_1,R_2,\ldots,R_h)$ obeying (i), (ii) and (iii) is transformed into a stack~$\overline{\S} = (\ovR_1,\ovR_2,
\ldots,\ovR_{\ovh})$
\begin{itemize}
\item by inserting a run,~$\ovh = h+1$ and~$\ovr_{\ovh-2} = r_{h-1} \geqslant \alpha r_h = \alpha \ovr_{\ovh-1}$;

\item by merging the runs~$R_{h-1}$ and~$R_h$, we have~$\ovh = h-1$ and~$\ovr_{\ovh-2} = r_{h-3} \geqslant \alpha r_{h-2} = \alpha \ovr_{\ovh-1}$;

\item by merging the runs~$R_{h-2}$ and~$R_{h-1}$, case \#1 just occurred and~$\ovh = h-1$, so that
\[\min\{\ovr_{\ovh-2},\alpha\ovr_{\ovh}\} =
\min\{r_{h-3},\alpha r_h\} \geqslant \alpha r_{h-2}
\geqslant (\alpha-1)(r_{h-2} + r_{h-1}) = (\alpha-1) \ovr_{\ovh-1}.\]
\end{itemize}

In each case,~$\overline{\S}$ satisfies (ii) and (iii), which completes the induction and the proof.
\end{proof}

Lemmas~\ref{lem:aMS::right} and~\ref{lem:aMS::left} focus on inequalities involving the lengths of a given run~$R$ belonging to a merge tree induced by \aMS, and of its ancestors.
In each case, the stack just before the run~$R$ is merged is denoted by~$\S = (R_1,R_2,\ldots,R_h)$.
In what follows, we also set~$\alpha^\star = \min\{\alpha,1+1/\alpha,1+(\alpha-1)/\alpha\}$.

\begin{lemma}\label{lem:aMS::right}
If~$\aR{1}$ is a right run,~$\ar{2} \geqslant \alpha^\star r$.
\end{lemma}

\begin{proof}
Let~$i$ be the integer such that~$R = R_i$.
If~$i = h-2$, (i) shows that~$r_{i-1} \geqslant \alpha r_i$.
If~$i = h-1$, (ii) shows that~$r_{i-1} \geqslant (\alpha-1) r_i$.
In both cases,~$R_{i-1}$
descends from~$\aR{2}$, and
thus~$\ar{2} \geqslant r + r_{i-1}
\geqslant \alpha r$.
Finally, if~$i = h$, the run~$R$ is a \#2 or \#3-right run,
which means that~$r_{h-2} \geqslant r$ and that~$R_{h-2}$
descends from the left sibling of~$\ar{1}$.
It follows that~$\ar{2} \geqslant r + r_{h-2} \geqslant 2 r \geqslant
(1 + 1/\alpha) r$.
\end{proof}

\begin{lemma}\label{lem:aMS::left}
If~$R$ is a left run,~$\ar{2} \geqslant \alpha^\star r$.
\end{lemma}

\begin{proof}
We treat three cases independently, depending on
whether~$R$ is a \#1-, a \#2 or a \#3-left run.
In each case, we assume that~$\aR{1}$ is a left run,
since Lemma~\ref{lem:aMS::right} already proves that~$\ar{2} \geqslant \alpha^\star r$ when~$\aR{1}$ is a right
run.
\begin{itemize}
\item If~$R$ is a \#1-left run, the run~$R = R_{h-2}$ is merged with~$R_{h-1}$ and~$r_{h-2} < r_h$.
Since~$\aR{1}$ is a left run,~$R_h$ descends from~$\aR{2}$, and thus~$\ar{2} \geqslant r + r_h \geqslant 2r \geqslant (1 + 1/\alpha) r$.

\item If~$R$ is a \#2-left run, the run~$R = R_{h-1}$ is merged with~$R_h$ and~$r < \alpha r_h$.
It follows, in that case, that~$\ar{2} \geqslant \ar{1} = r + r_h \geqslant (1 + 1/\alpha) r$.

\item If~$R$ is a \#3-left run, the run~$R = R_{h-1}$ is merged with~$R_h$ and~$r_{h-2} < \alpha r_{h-1}$.
Hence, (iii) proves that~$(\alpha-1) r \leqslant \alpha r_h$, so that~$\ar{2} \geqslant \ar{1} = r + r_h \geqslant (1 + (\alpha-1)/\alpha) r$. \qedhere
\end{itemize}
\end{proof}

\begin{proposition}\label{pro:fast-growth-aMS}
The algorithm \aMS has the fast-growth property.
\end{proposition}

\begin{proof}
Let~$\T$ be a merge tree induced by \aMS, and let~$R$ be a run in~$\T$.
Lemma~\ref{lem:aMS::right} shows that~$\ar{3} \geqslant \ar{2} \geqslant \alpha^\ast r$ if~$\aR{1}$ is a right run, and Lemma~\ref{lem:aMS::left} shows that~$\ar{3} \geqslant \alpha^\ast \ar{1} \geqslant \alpha^\ast r$ if~$\aR{1}$ is a left run anyway.
\end{proof}

\subsection{Algorithms with the tight middle-growth property}
\label{subsec:other-tight}

\subsubsection{\PoS}
\label{subsubsec:PoS:2}

\begin{proposition}
\label{pro:POS:tight}
The algorithm \PoS has the tight middle-growth property.
\end{proposition}

\begin{proof}
Let~$\T$ be a merge tree induced by \PoS and let~$R$ be a run in~$\T$ at depth at least~$h$.
We will prove that~$\ar{h} \geqslant 2^{h-4}$.

If~$h \leqslant 4$, the desired inequality is immediate.
Then, if~$h \geqslant 5$, let~$n$ be the length of the array on which~$\T$ is induced, and let~$p$ and~$\ap{h-2}$ be the respective powers of the runs~$R$ and~$\aR{h-2}$.
Definition~\ref{def:PoS} and Lemma~\ref{lem:PoS:power} prove that~$2^{\ap{h-2}+h-4} \leqslant 2^{p-2} \leqslant 2^{p-2} r < n < 2^{\ap{h-2}} \ar{h}$.
\end{proof}

\subsubsection{\NMS}
\label{subsubsec:NMS}

The algorithm \NMS consists in a plain binary merge sort, whose unit pieces of data to be merged are runs instead of being single elements.
Thus, we identify \NMS with the fundamental property that describes those merge trees it induces.

\begin{definition}\label{def:tree:NMS}
Let~$\T$ be a merge tree induced by \NMS, and let~$R$ and~$\ovR$ be two runs that are siblings of each other in~$\T$.
Denoting by~$n$ and~$\overline{n}$ the respective numbers of leaves of~$\T$ that descend from~$R$ and from~$\ovR$, we have~$|n - \overline{n}| \leqslant 1$.
\end{definition}

\begin{proposition}\label{pro:fast-growth-NMS}
The algorithm \NMS has the tight middle-growth property.
\end{proposition}

\begin{proof}
Let~$\T$ be a merge tree induced by \NMS, let~$R$ be a run in~$\T$, and let~$h$ be its height.
We will prove by induction on~$h$ that, if~$h \geqslant 1$, the run~$R$ is an ancestor of at least~$2^{h-1}+1$ leaves of~$\T$, thereby showing that~$r \geqslant 2^{h-1}$.

First, this is the case if~$h = 1$.
Then, if~$h \geqslant 2$, let~$R_1$ and~$R_2$ be the two children of~$R$.
One of them, say~$R_1$, has height~$h-1$.
Let~$n$,~$n_1$ and~$n_2$ be the numbers of leaves that descend from~$R$,~$R_1$ and~$R_2$, respectively.
The induction hypothesis shows that
\[n = n_1 + n_2 \geqslant 2 n_1-1 \geqslant
2 \times (2^{h-2}+1) - 1 = 2^{h-1}+1,\]
which completes the proof.
\end{proof}

\subsubsection{\ShS}
\label{subsubsec:ShS}

The algorithm \ShS is presented in Algorithm~\ref{alg:ShS}.
Like \cASS, it relies on the notion of \emph{level} of a run.

\begin{algorithm}[ht]
\begin{small}
\SetArgSty{texttt}
\DontPrintSemicolon
\Input{Array~$A$ to sort}
\Result{The array~$A$ is sorted into a single run.
That run remains on the 
stack.}
\Note{The height of the stack~$\S$ is denoted by~$h$; its~$i$\textsuperscript{th} deepest run by~$R_i$; the length of~$R_i$~by~$r_i$.
Finally, we set~$\ell_i = \lfloor \log_2(r_i) \rfloor$.
When two consecutive runs of~$\S$ are merged, they are replaced, in~$\S$, with the run resulting from the merge. \justifying}
\BlankLine {}$\S \gets~$ an empty stack

\While(\label{main-loop:ShS:start}){\true}{
 \If(\label{test:ShS:case-1})
 {\textrm{$h \geqslant 1$ and~$\ell_{h-1}
 \leqslant \ell_h$}}
 {merge the runs~$R_{h-1}$ and~$R_h$\label{ShS:std:merge}}
 \ElseIf{\textrm{the end of the array has not yet been reached}}
 {find a new monotonic run~$R$, make it non-decreasing,
 and push it onto~$\S$\label{ShS:case-push}}
 \Else{break\label{algline:ShS:end_inner_loop}}
}
\While{$h \geqslant 2$}{
 merge the runs~$R_{h-1}$ and~$R_h$
 \label{alg:ShS:collapse}
}
\end{small}
\caption{\ShS\label{alg:ShS}}
\end{algorithm}

Our proof is based on the following result, which appears in the proof of~\cite[Theorem 11]{BuKno18}.

\begin{lemma}\label{lemma:ShS:invariant}
Let~$\S = (R_1,R_2,\ldots,R_h)$ be a stack obtained while executing \ShS.
We have~$\ell_i \geqslant \ell_{i+1}+1$ whenever~$1 \leqslant i \leqslant h-2$.
\end{lemma}

Lemmas~\ref{lem:ShS:right} and~\ref{lem:ShS:left-right} focus on inequalities involving the lengths of a given run~$R$ belonging to a merge tree induced by \cASS, and of its ancestors.
In each case, the stack just before the run~$R$ is merged is denoted by~$\S = (R_1,R_2,\ldots,R_h)$.

However, and unlike Sections~\ref{subsubsec:cASS} to~\ref{subsubsec:aMS}, we cannot simulate the merge operations that occur in line~\ref{alg:ShS:collapse} as if they had occurred in line~\ref{ShS:std:merge}.
Instead, we say that a run~$R$ is \emph{rightful} if~$R$ and its ancestors are all right runs (i.e., if~$R$ belongs to the rightmost branch of the merge tree), and that~$R$ is \emph{standard} otherwise.

\begin{lemma}\label{lem:ShS:right}
Let~$R$ be a run in~$\T$, and let~$k \geqslant 1$ be an integer.
If each of the~$k$ runs~$\aR{0},\ldots,\aR{k-1}$ is a right run,~$\al{k} \geqslant k-1$.
\end{lemma}

\begin{proof}
For all~$i \leqslant k$, let~$u(i)$ be the least integer such that~$R_{u(i)}$ descends from~$\aR{i}$.
Since the~$k$ runs~$\aR{0},\ldots,\aR{k-1}$ are right runs,~$u(k) < u(k-1) < \ldots < u(0) = h$.
Thus, Lemma~\ref{lemma:ShS:invariant} proves~that
\[\al{k} \geqslant \ell_{u(k)} \geqslant
\ell_{u(1)} + (k-1) \geqslant k-1. \qedhere\]
\end{proof}

\begin{lemma}\label{lem:ShS:left-right}
Let~$R$ be a run in~$\T$, and let~$k \geqslant 1$ be an integer.
If~$R$ is a left run and the~$k-1$ runs~$\aR{1},\ldots,\aR{k-1}$ are right runs, we have~$\al{k} \geqslant \ell+k-1$ if~$\aR{1}$ is a rightful run, and~$\al{k} \geqslant \ell+k$ if~$\aR{1}$ is a standard run.
\end{lemma}

\begin{proof}
First, assume that~$k = 1$.
If~$\aR{1}$ is rightful, the desired inequality is immediate.
If~$\aR{1}$ is standard, however, the left run~$R = R_{h-1}$ was merged with the run~$R_h$ because~$\ell \leqslant \ell_h$.
In that case, it follows that~$\ar{1} = r + r_h \geqslant 2^\ell + 2^{\ell_h} \geqslant 2^\ell + 2^\ell = 2^{\ell+1}$, i.e., that~$\al{1} \geqslant \ell+1$.

Assume now that~$k \geqslant 2$.
Note that~$\aR{k}$ is rightful if and only if~$\aR{1}$ is also rightful.
Then, for all~$i \leqslant k$, let~$u(i)$ be the least integer such that the run~$R_{u(i)}$ descends from~$\aR{i}$.
Since~$\aR{1},\ldots,\aR{k-1}$ are right runs,~$u(k) < u(k-1) < \ldots < u(1) = h-1$.

In particular, let~$R'$ be the left sibling of~$\aR{k-1}$:
this is an ancestor of~$R_{u(k)}$, and the left child of~$\aR{k}$.
Consequently, Lemma~\ref{lemma:ShS:invariant} and applying our study of the case~$k = 1$ to the run~$R'$ conjointly prove that
\begin{itemize}
\item {}$\al{k} \geqslant \ell' \geqslant \ell_{u(k)} \geqslant \ell_{u(1)}+k-1 = \ell+k-1$ if~$\aR{1}$ and~$\aR{k}$ are rightful;
\item {}$\al{k} \geqslant \ell'+1 \geqslant \ell_{u(k)}+1 \geqslant \ell_{u(1)}+k = \ell+k$  if~$\aR{1}$ and~$\aR{k}$ are standard. \qedhere
\end{itemize}
\end{proof}

\begin{proposition}\label{pro:middle-growth-ShS}
The algorithm \ShS has the tight middle-growth property.
\end{proposition}

\begin{proof}
Let~$\T$ be a merge tree induced by \ShS and let~$R$ be a run in~$\T$ at depth at least~$h$.
Let~$a_1 < a_2 < \ldots < a_k$ be the non-negative integers smaller than~$h$ for which~$\aR{a_i}$ is a left run.
We also set~$a_{k+1} = h$.
Lemma~\ref{lem:ShS:right} proves that~$\al{a_1} \geqslant a_1-1$.
Then, for all~$i < k$, the run~$\aR{a_i+1}$ is standard, since it descends from the left run~$\aR{a_k}$, and thus Lemma~\ref{lem:ShS:left-right} proves that~$\al{a_{i+1}} \geqslant \al{a_i} + a_{i+1} - a_i$. 
Lemma~\ref{lem:ShS:left-right} also proves that~$\al{a_{k+1}} \geqslant \al{a_k} + a_{k+1} - a_k-1$.
It follows that~$\al{h} = \al{a_{k+1}} \geqslant h-2$, and therefore that~$\ar{h} \geqslant 2^{\al{h}} \geqslant 2^{h-2}$.
\end{proof}

\subsection{Algorithms with the middle-growth property}
\label{subsec:other-middle}

\subsubsection{\aSS}
\label{subsubsec:aSS}

The algorithm \aSS, which
predated and inspired its variant
\aMS, is presented in
Algorithm~\ref{alg:aSS}.

\begin{algorithm}[h]
\begin{small}
\SetArgSty{texttt}
\DontPrintSemicolon
\Input{Array~$A$ to sort, parameter~$\alpha > 1$}
\Result{The array~$A$ is sorted into a single run.
That run remains on the stack.}
\Note{The height of the stack~$\S$ is denoted by~$h$; its~$i$\textsuperscript{th} deepest run by~$R_i$; the length of~$R_i$~by~$r_i$.
When two consecutive runs of~$\S$ are merged, they are replaced, in~$\S$, with the run resulting from the merge. \justifying}
\BlankLine {}$\S \gets~$ an empty stack

\While(\label{aSS:main-loop:start}){\true}{
 \If{\textrm{$h \geqslant 2$ and 
~$r_{h-1} \leqslant \alpha r_h$}}
 {merge the runs~$R_{h-1}$ and~$R_h$}
 \ElseIf{\textrm{the end of the array has not yet been reached}}
 {find a new monotonic run~$R$, make it non-decreasing, and push it onto~$\S$\label{aSS:case-push}}
 \Else{break\label{algline:aSS:end_inner_loop}}
}
\While{$h \geqslant 2$}{
 merge the runs~$R_{h-1}$ and~$R_h$
 \label{alg:aSS:collapse}
}
\end{small}
\caption{\aSS\label{alg:aSS}}
\end{algorithm}

Our proof is based on the following result, which appears in~\cite[Lemma 2]{AuNiPi15}.

\begin{lemma}\label{lemma:aSS:invariant}
Let~$\S = (R_1,R_2,\ldots,R_h)$ be a stack obtained while executing \mbox{\aSS}.
We have~$r_i > \alpha r_{i+1}$ whenever~$1 \leqslant i \leqslant h-2$.
\end{lemma}

Lemmas~\ref{lem:aSS::right} and~\ref{lem:aSS::left} focus on inequalities involving the lengths of a given run~$R$ belonging to a merge tree induced by \aSS, and of its ancestors.
In each case, the stack just before the run~$R$ is merged is denoted by~$\S = (R_1,R_2,\ldots,R_h)$.
Furthermore, like in Section~\ref{subsubsec:ShS}, we say that a run~$R$ is \emph{rightful} if~$R$ and its ancestors are all right runs, and that~$R$ is \emph{standard} otherwise.

\begin{lemma}\label{lem:aSS::right}
Let~$R$ be a run in~$\T$, and let~$k$ and~$m$ be two integers.
If~$k$ of the~$m+1$ runs~$\aR{0}, \aR{1}, \ldots, \aR{m}$ are right runs,~$\ar{m+1} \geqslant \alpha^{k-1}$.
\end{lemma}

\begin{proof}
Let~$a_1 < a_2 < \ldots < a_k$ be the~$k$ smallest integers such that~$\aR{a_1},\aR{a_2}, \ldots,\aR{a_k}$ are right runs.
For all~$i \leqslant k-1$, let~$\aS{i}$ be the left sibling of~$\aR{a_i}$, and let~$u(i)$ be the least integer such that~$R_{u(i)}$ descends from~$\aS{i}$.
Since~$R_{u(i)}$ descends from~$\aR{a_{i+1}}$, we know that~$u(i+1) < u(i)$, so that~$u(k) \leqslant u(1) - (k-1)$.
Thus, Lemma~\ref{lemma:aSS:invariant} proves that~$\ar{m+1} \geqslant r_{u(k)} \geqslant \alpha^{k-1} r_{u(1)} \geqslant \alpha^{k-1}$.
\end{proof}

\begin{lemma}\label{lem:aSS::left}
Let~$R$ be a run in~$\T$.
If~$R$ is a left run and if$~\aR{1}$ is a standard run,~$\ar{1} \geqslant (1+1/\alpha)r$.
\end{lemma}

\begin{proof}
The left child of~$\aR{1}$ coincides with~$R = R_{h-1}$ and, since~$\aR{1}$ is a standard run,~$r_{h-1} \leqslant \alpha r_h$.
It follows immediately that~$\ar{1} = r_{h-1} + r_h \geqslant (1+1/\alpha) r$.
\end{proof}

\begin{proposition}\label{pro:middle-growth-aSS}
The algorithm \aSS has the middle-growth property.
\end{proposition}

\begin{proof}
Let~$\T$ be a merge tree induced by \ShS, let~$R$ be a run in~$\T$ at depth at least~$h$, and let~$\beta = \min\{\alpha,1+1/\alpha\}^{1/4}$.
We will prove that~$\ar{h} \geqslant \beta^h$.

Indeed, let~$k = \lceil h/2 \rceil$.
We distinguish two cases, which are not mutually exclusive if~$h$ is even.
\begin{itemize}
\item If at least~$k$ of the~$h$ runs~$R, \aR{1},\aR{2}, \ldots, \aR{h-1}$ are right runs, Lemma~\ref{lem:aSS::right} shows that
\[\ar{h} \geqslant \alpha^{k-1} \geqslant \beta^{4(k-1)}.\]
\item If at least~$k$ of the~$h$ runs~$R, \aR{1},\aR{2}, \ldots, \aR{h-1}$ are left runs, let~$a_1 < a_2 < \ldots < a_k$ be the~$k$ smallest integers such that~$\aR{a_1},\aR{a_2}, \ldots, \aR{a_k}$ are left runs.
When~$j < k$, the run~$\aR{a_j+1}$ is standard, since it descends from the left run~$\aR{a_k}$.
Therefore, due to Lemma~\ref{lem:aSS::left}, an immediate induction on~$j$ shows that~$\ar{a_j+1} \geqslant (1+1/\alpha)^j$ for all~$j \leqslant k-1$.
It follows that
\[\ar{h} \geqslant \ar{a_{k-1}+1} \geqslant
(1+1/\alpha)^{k-1} \geqslant \beta^{4(k-1)}.\]
\end{itemize}

Consequently,~$\ar{h} \geqslant \beta^{4(k-1)} \geqslant \beta^{2h-4} \geqslant \beta^h$ when~$h \geqslant 4$, whereas~$\ar{h} \geqslant 1 = \beta^h$ when~$h = 0$ and~$\ar{h} \geqslant 2 \geqslant 1+1/\alpha \geqslant \beta^h$ when~$1 \leqslant h \leqslant 3$.
\end{proof}

\section{\TS's update policy}
\label{sec:TS-update}

In this Section, we study \TS's merging routine as it was implemented in~\cite{Java,Python,Octave,V8}.
This routine aims at merging consecutive non-decreasing runs~$A$ and~$B$.
Although similar to the description we gave in Section~\ref{sec:description}, this implementation does not explicitly rely on~$\bt$-galloping, but only on~$0$-galloping, and it includes complicated side-cases or puzzling low-level design choices that we neglected on purpose when proposing a simpler notion of~$\bt$-galloping.

For instance, one may be surprised to observe that the number of consecutive elements coming from~$A$ (or~$B$) needed to trigger the galloping mode depends on whether elements from a previous dual run were discovered via galloping or naïvely.
Another peculiar decision is that, when the galloping mode is triggered by a long streak of elements coming \emph{from~$B$}, galloping is first used to count subsequent elements coming \emph{from~$A$}; fortunately, it is then also used to count elements coming from~$B$, which approximately amounts to using~$1$-galloping instead of~$0$-galloping.
Finally, instances of the galloping mode are launched in pairs, and the algorithm infers that galloping was efficient as soon as one of these two instances was efficient; this makes the update rule easy to fool.

\begin{algorithm}[p]
\begin{small}
\SetArgSty{texttt}
\SetKwProg{Pn}{Function}{:}{}
\SetKwFunction{LRMerge}{\textsf{LRMerge}}
\SetKwFunction{RLMerge}{\textsf{RLMerge}}
\SetKwFunction{LRMergeStep}{\textsf{MergeStep}$(a,b)$}
\SetKwFunction{LRNaive}{\textsf{Naïve}$(a,b,t_1,t_2)$}
\SetKwFunction{LRGallop}{\textsf{Gallop}$(a,b)$}
\DontPrintSemicolon

\Input{Non-decreasing runs~$A$ and~$B$ to merge, of lengths~$a$ and~$b$}
\Result{The runs~$A$ and~$B$ are merged into a single run~$C$.}
\Note{Let~$S_1,S_2,\ldots,S_\sigma$ be the non-decreasing dual runs of the sub-array spanned by~$A$ and~$B$.
The run~$A$ (resp.,~$B$) contains~$a_i$ (resp.,~$b_i$) elements whose value belongs to~$S_i$.\justifying}

\justifying\BlankLine
use~$0$-galloping to discover~$a_1$ and~$b_\sigma$\label{alg-6:line-1}\;
\textsf{Status}~$\gets$ \textsf{naïve}\;
\lIf{$a - a_1 \leqslant b - b_\sigma$}
 {\textsf{LRMerge}
 \tcp*[f]{$C$ is discovered from left to right\label{alg:TS:LR}}}
\lElse{\muteelse{$a - a_1 \leqslant b - b_\sigma$}%
 \textsf{RLMerge}
\tcp*[f]{$C$ is discovered from right to left\label{alg:TS:RL}}}

\drawline
\setcounter{AlgoLine}{4}
\Pn{\LRMerge}{
\textsf{Naïve}$(0,b_1,\mathbf{t},\mathbf{t}+1)$\;
\lFor{$i=2,3,\ldots,\sigma-2$}{\textsf{MergeStep}$(a_i,b_i)$}
\lIf{\textrm{$\sigma \geqslant 3$ and~$a_\sigma \geqslant 2$}}{\textsf{MergeStep}$(a_{\sigma-1},b_{\sigma-1})$\label{alg-6:line-8}}
\lElseIf{$\sigma \geqslant 3$}{\muteelseif{$\sigma \geqslant 3$ and~$a_\sigma \geqslant 2$}{$\sigma \geqslant 3$}%
 \textsf{MergeStep}$(a_{\sigma-1},0)$}}

\drawline
\setcounter{AlgoLine}{9}
\Pn{\RLMerge}{
\textsf{Naïve}$(a_\sigma,b_{\sigma-1},\mathbf{t}+1,\mathbf{t}+1)$\;
\lFor{$i=2,3,\ldots,\sigma-2$}{\textsf{MergeStep}$(a_{\sigma+1-i},b_{\sigma-i})$}
\lIf{\textrm{$\sigma \geqslant 3$ and~$a_2 \geqslant 2$}}{\textsf{MergeStep}$(a_2,0)$\label{alg-6:line-14}}}

\def\myCondition{\textrm{\textsf{Status}~$=$ \textsf{fail}}}%
\drawline
\setcounter{AlgoLine}{13}
\Pn{\LRMergeStep}{
 \lIf{\textsf{Status}~$=$ \textsf{fail}}{\muteif{\textsf{Status}~$=$ \textsf{fail}}{\textsf{Status}~$=$ \textsf{success}}%
$\mathbf{t} \gets \max\{2,\mathbf{t}+1\}$}
 \lElseIf{\textsf{Status}~$=$ \textsf{success}}{$\mathbf{t} \gets \mathbf{t}-1$}
 \lIf{\textsf{Status}~$=$ \textsf{naïve}}{\muteif{\textsf{Status}~$=$ \textsf{naïve}}{\textsf{Status}~$=$ \textsf{start}}%
 \textsf{Naïve}$(a,b,\mathbf{t},\mathbf{t})$}
 \lElseIf{\textsf{Status}~$=$ \textsf{fail}}{\muteself{\textsf{fail}}{\textsf{start}}\textsf{Naïve}$(a,b,\mathbf{t}+1,\mathbf{t})$}
 \lElseIf{\textsf{Status}~$=$ \textsf{start}}{\textsf{Gallop}$(a+1,b)$}
 \lElse{\elseifmute{\textsf{Status}~$=$ \textsf{start}~}%
  \textsf{Gallop}$(a,b)$}}

\def\myCondition{\textrm{$a \geqslant t_1+7$ or ($a \geqslant t_1$ and~$b \geqslant 8$) or~$b \geqslant t_2+8$}}%
\drawline
\setcounter{AlgoLine}{20}
\Pn{\LRNaive}{
use~$t_1$-galloping to discover~$a$\;
\lIf{$a \geqslant t_1$}{use~$0$-galloping to discover~$b$}
\lElse{\muteelse{$a \geqslant t_1$}%
 use~$t_2$-galloping to discover~$b$\label{alg-6:line-24}}
\lIf{\myCondition}{\textsf{Status}~$\gets$ \textsf{success}}
\lElseIf{\textrm{$a \geqslant t_1$ or~$b \geqslant t_2+1$}}{\muteelseif{\myCondition}{$a \geqslant t_1$ or~$b \geqslant t_2+1$}%
\textsf{Status}~$\gets$ \textsf{fail}}
\lElseIf{$b = t_2$}{\muteelseif{\myCondition}{$b = t_2$}%
\textsf{Status}~$\gets$ \textsf{start}\;
\muteelseif{\myCondition\textbf{else if : }}{}%
cancel the last comparison}
\lElse{\muteelse{\myCondition}%
 \textsf{Status}~$\gets$ \textsf{naïve}}}

\drawline
\setcounter{AlgoLine}{29}
\Pn{\LRGallop}{
use~$0$-galloping to discover~$a$ and~$b$\;
\lIf{\textrm{$a \geqslant 8$ or~$b \geqslant 8$}}{\textsf{Status}~$\gets$ \textsf{success}}
\lElse{\muteelse{$a \geqslant 8$ or~$b \geqslant 8$}%
 \textsf{Status}~$\gets$ \textsf{fail}}}
\end{small}
\bigskip
\caption{\TS's merging routine\label{alg:TS-merge}.
This routine is best modelled by using a variable \textsf{Status} that may take four values: \textsf{naïve} if launching a~$\mathbf{t}$-gallop (i.e., starting naïvely) is requested; \textsf{success} if the last gallop was successful (i.e., saved comparisons); \textsf{fail} if the last gallop was unsuccessful (i.e., wasted comparisons); \textsf{start} if a~$0$-gallop is to be launched.
A few side cases, having little influence on the number of comparisons but not on the dynamics of~$\mathbf{t}$, are overlooked here and discussed in Section~\ref{sec:TS-routine-caveat}.}
\end{algorithm}

Algorithm~\ref{alg:TS-merge} consists in a pseudo-code transcription of this merging routine, where we focus only on the number of comparisons performed.
Thus, we disregard both the order in which comparisons are performed and element moves.
For the sake of simplicity,\footnote{Counting precisely comparisons performed by \TS's merging routine required 166 Java code lines~\cite{JuCode23}.} we also deliberately simplified the behaviour of \TS's routine, by overlooking some side cases, listed in Section~\ref{sec:TS-routine-caveat} below.

Algorithm~\ref{alg:TS-merge} explicitly relies on the values of the (global) parameter~$\bt$ and on a (global) \emph{status} variable~$\mathsf{S}$, implicit in the implementations of \TS~\cite{Java,Python,Octave,V8}.
In practice, the parameter~$\mathbf{t}$ is often initially set to~$\bt = 7$, although other initial values might be chosen; it is subsequently updated by each instance of the merging routine, and is \emph{not} reset to~$\bt = 7$ between two such instances.
In what follows, we just note~$\btinit$ the initial value of~$\bt$.

The two main reasons behind introducing this merging routine are as follows:
(i)~without incurring a super-linear overhead to do so, one wants to (ii)~merge arrays with a bounded number of values in linear time.
Below, we provide counter-examples to both these assertions.

\subsection{Caveat}
\label{sec:TS-routine-caveat}

Providing counter-examples to the goals~(i) and~(ii) requires paying attention to some low-level details we might wish to skip.
The first one is that, as indicated in Section~\ref{sec:intro}, when decomposing an array into runs, \TS makes sure that these runs are \emph{long enough}, and extends them if they are too short.
In practice, given a size~$\ms$ that is considered large enough in a given implementation, both our counter-examples consist only of runs of length at least~$\ms$.
Such a size may vary between two implementations (from~$32$ in Java to~$64$ in Python), which is why we kept it as a parameter.

The second one is that Algorithm~\ref{alg:TS-merge} represents a \emph{simplified} version of \TS's actual merging routine.
Here are the differences between Algorithm~\ref{alg:TS-merge} and that routine:
\begin{enumerate}
\item The first element of~$A$ in~$S_2$ and the last element of~$B$ in~$S_{\sigma-1}$ have already been identified in line~\ref{alg-6:line-1}.
Thus, \TS's routine tries to save one comparison per merge by not rediscovering these elements naïvely during the last calls to \textsf{MergeStep} in lines~\ref{alg-6:line-8} or~\ref{alg-6:line-14}.
Nothing is changed if these elements are discovered through galloping. 
\item In line~\ref{alg-6:line-24}, \TS's routine actually uses~$1$-galloping instead of~$0$-galloping, because it loses one step trying to discover a streak of~$0$ elements from~$A$.
\item As mentioned in Section~\ref{sec:description}, \TS's galloping routine may actually go faster than forecast when discovering a streak of~$x$ elements from a run in which less than~$2x$ elements remain to be discovered.
\end{enumerate}

In practice, the first difference is harmless for our study: it concerns only a constant number of comparisons per run merge, i.e.,~$\mathcal{O}(n)$ comparisons in total when sorting an array of length~$n$.
The second difference is also harmless, because it just makes \TS's routine \emph{more expensive} than what we will claim below.
The third difference is tackled in two different ways:
in Section~\ref{sec:TS-routine-overhead}, we simply avoid falling in that case, and in Section~\ref{sec:TS-routine-n-log-n}, we will brutally take~$0$ as an under-estimation of the merge cost of a galloping phase --- an under-estimation that remains valid regardless of how the galloping is carried.

Finally, and although we are referring to \emph{\TS}'s merging routine, we would also like to prove that this routine suffers the same shortcomings with many relevant algorithms.
Thus, we made sure that all algorithms would induce the same merge tree on our counterexample, thereby having the same behaviour.

\subsection{Super-linear overhead}
\label{sec:TS-routine-overhead}

One desirable safeguard, when using \TS's routine instead of a naïve merging routine, is that the overhead incurred by this change should be small.
In other words, the user should never be heavily penalised for using \TS's routine.

More precisely, when using an algorithm~$\A$ to sort an array~$A$ of length~$n$, let~$\mathsf{U}_A$ be the number of comparisons performed by~$\A$ if it relies on a naïve merging routine, and let~$\mathsf{V}_A$ be the number of comparisons performed by~$\A$ if it uses \TS's routine.
Given that the \emph{worst-case} complexity of a reasonable sorting algorithm is~$\O(n \log(n))$, one may reasonably expect inequalities of the form~$\mathsf{V}_A \leqslant \mathsf{U}_A + o(n \log(n))$.
Below, we prove that such inequalities are invalid, meaning that using \TS's routine may cost a positive fraction of the total merge cost of the algorithm, even in cases where this merge cost is already quite bad.

\begin{proposition}
\label{pro:super-linear-overhead}
Let~$\A$ be one of the algorithms studied in Section~\ref{sec:pos-fast-growth}.
Let~$\mathsf{U}_A$ be the number of comparisons performed by~$\A$ to sort an array~$A$ when using a naïve merging routine, and let~$\mathsf{V}_A$ be the number of comparisons performed by~$\A$ when using \TS's routine.
There exists an array~$A$, whose length~$n$ may be arbitrarily large, for which~$\mathsf{V}_A - \mathsf{U}_A \in \Omega(n \log(n))$.
\end{proposition}

\begin{proof}
What dictates the dynamics of Algorithm~\ref{alg:TS-merge}, when merging runs~$A$ and~$B$ into a new run~$C$, is not the precise values of those elements of~$A$ and~$B$, but only the lengths of the  consecutive streaks of elements of~$C$ coming either from~$A$ or from~$B$.
Thus, we only need to control these lengths.

Based on this remark, we use two building bricks.
The first brick, represented in Figure~\ref{fig:TS-superlinear}~(left), aims at changing~$\bt$, initially equal to~$\btinit$, to let it reach the value~$\bt = 5$;
it is used once, as the first merge performed in \TS.
The second brick, represented in Figure~\ref{fig:TS-superlinear}~(right), aims at maximising the number of comparisons performed, without changing the value of~$\bt$: if the merge started with~$\bt = 5$, it shall end with~$\bt = 5$.

Our first brick is built by merging two runs~$R$ and~$R'$ with a given length~$\ell \geqslant 8\btinit+18$, as indicated in Figure~\ref{fig:TS-superlinear}~(left).
When~$R$ and~$R'$ are being merged, both lengths~$r_1$ and~$r'_\sigma$ are zero, and the function \textsf{LRMerge} is called.
First encountering~$\btinit+8$ consecutive elements from~$R$ triggers a call to \textsf{Gallop}; this function is successfully called~$\btinit$ times, which brings down~$\bt$ to~$0$.
Then, the nine following elements from~$R$ (and nine elements from~$R'$) result in unsuccessfully calling \textsf{Gallop} several times, which gradually raises up~$\bt$ to~$5$.
The padding area has no influence on~$\bt$;
its only purpose is to ensure that~$R$ and~$R'$ have the desired length~$\ell$.

Our second brick is built by merging two runs~$S$ and~$S'$ with a given length~$11m$, as indicated in Figure~\ref{fig:TS-superlinear}~(right).
Provided that the parameter~$\bt$ has been set to~$5$, encountering~$8$ consecutive elements coming from~$S$ triggers a call to \textsf{Gallop}.
This call is successful, because it helps us to discover~$8$ consecutive elements from~$S'$, and it is immediately followed by an unsuccessful call to \textsf{Gallop}, which only helps us to discover streaks of length~$3$.
After these two calls,~$\bt$ is still equal to~$5$, which allows us to repeat such sequences of calls.
Hence, Algorithm~\ref{alg:TS-merge} uses~$23$ comparisons to discover each block~$8+3+8+3$ elements --- one more comparison than needed.
Thus, Algorithm~\ref{alg:TS-merge} uses a total of~$23m-3$ element comparisons to merge two runs~$S$ and~$S'$ of length~$11m$, which is more than the~$22m-1$ comparisons that a naïve routine would use.

\begin{figure}[t]
\begin{center}
\begin{tikzpicture}[xscale=0.345,yscale=0.32]
\foreach \x/\w/\a/\b in {0/2/0/1,2/8/\btinit+8/8,10/2/7/8,%
12/2/7/8,18/2/7/8,20/2/1/1,22/2/2/2,24/2/4/4,26/2/2/2,%
28/2/1/1,30/2/1/1,36/2/1/1,38/2/1/0}{
\draw[thick] (0.5*\x+0.1,4.4) --++ (0.5*\w-0.2,0) (0.5*\x+0.1,2.3) --++ (0.5*\w-0.2,0);
\draw[thick,->,>=stealth] (0.5*\x+0.25*\w,4.4) --++ (0,-0.3) --++ (-0.125*\w,0) --++ (0,-0.75);
\draw[thick,->,>=stealth] (0.5*\x+0.25*\w,2.3) --++ (0,0.3) --++ (0.125*\w,0) --++ (0,0.75);
 \node at (0.5*\x+0.25*\w,5.7) {$\a$};
 \node at (0.5*\x+0.25*\w,1) {$\b$};
 \draw (0.5*\x,0) --++ (0,2) ++ (0,2.7) --++ (0,2);
}
\foreach \x in {7,16}{
 \foreach \y in {1,3.35,5.7} {\node at (\x+1.1,\y) {$\cdots$};}
 \draw (\x,0) --++ (0,2) ++ (0,2.7) --++ (0,2);
}

\draw[very thick,<->,>=stealth] (5,7.4) --++ (5,0);
\draw[very thick,<->,>=stealth] (14,7.4) --++ (5,0);
\node[anchor=south] at (7.5,7.4) {$\btinit$ values};
\node[anchor=south] at (16.5,7.4) {\emph{paddin\smash{g}}};

\draw[very thick] (0,4.7) --++ (20,0) --++ (0,2) --++ (-20,0) -- cycle;
\draw[very thick] (0,0) --++ (20,0) --++ (0,2) --++ (-20,0) -- cycle;
\draw[very thick] (0,3.35) --++ (7,0) ++ (2,0) --++ (7,0) ++ (2,0) --++ (2,0);

\node[anchor=east] at (-0.1,5.7) {Run \rlap{$R$}\phantom{$R'$}:};
\node[anchor=east] at (-0.1,3.35) {Run \rlap{$\overline{R}$}\phantom{$R'$}:};
\node[anchor=east] at (-0.1,1) {Run~$R'$:};

\begin{scope}[shift={(28,0)}]
\foreach \x/\w/\a/\b in {0/2/0/1,2/2/1/1,4/2/8/8,6/2/3/3,%
8/2/8/8,10/2/3/3,16/2/8/8,18/2/3/3,20/2/3/3,22/2/3/3,%
24/2/2/3,26/2/2/0}{
\draw[thick] (0.5*\x+0.1,4.4) --++ (0.5*\w-0.2,0) (0.5*\x+0.1,2.3) --++ (0.5*\w-0.2,0);
\draw[thick,->,>=stealth] (0.5*\x+0.25*\w,4.4) --++ (0,-0.3) --++ (-0.125*\w,0) --++ (0,-0.75);
\draw[thick,->,>=stealth] (0.5*\x+0.25*\w,2.3) --++ (0,0.3) --++ (0.125*\w,0) --++ (0,0.75);
 \node at (0.5*\x+0.25*\w,5.7) {$\a$};
 \node at (0.5*\x+0.25*\w,1) {$\b$};
 \draw (0.5*\x,0) --++ (0,2) ++ (0,2.7) --++ (0,2);
}
\foreach \x in {6}{
 \foreach \y in {1,3.35,5.7} {\node at (\x+1.1,\y) {$\cdots$};}
 \draw (\x,0) --++ (0,2) ++ (0,2.7) --++ (0,2);
}

\draw[very thick,<->,>=stealth] (2,7.4) --++ (8,0);

\draw[very thick,<->,>=stealth] (2,-0.7) --++ (2,0);
\draw[very thick,<->,>=stealth] (8,-0.7) --++ (2,0);
\node[anchor=south] at (6,7.4) {$m-1$ blocks};
\node[anchor=north] at (3,-0.7) {block};
\node[anchor=north] at (9,-0.7) {block};

\draw[very thick] (0,4.7) --++ (14,0) --++ (0,2) --++ (-14,0) -- cycle;
\draw[very thick] (0,0) --++ (14,0) --++ (0,2) --++ (-14,0) -- cycle;
\draw[very thick] (0,3.35) --++ (6,0) ++ (2,0) --++ (6,0);

\node[anchor=east] at (-0.1,5.7) {Run \rlap{$S$}\phantom{$S'$}:};
\node[anchor=east] at (-0.1,3.35) {Run \rlap{$\overline{S}$}\phantom{$S'$}:};
\node[anchor=east] at (-0.1,1) {Run~$S'$:};
\end{scope}
\end{tikzpicture}
\vspace{-1.7em}
\end{center}
\caption{Base bricks of our construction.
The diagrams should be read as follows. \newline
The run~$\overline{R}$ consists in~$1$ element from~$R'$, then~$\btinit+8$ elements from~$R$,~$8$ elements from~$R'$,
$7$~elements from~$R$,~$8$ elements from~$R'$,~$7$ elements from~$R$, \ldots, and finally~$1$ element from~$R$.
The runs~$R$ and~$R'$ both have a length~$\ell \geqslant 8\btinit+19$, and include a padding area of length~$\ell-(8\btinit+18)$.\newline
Similarly,~$\overline{S}$ consists in~$1$ element from~$S'$, then~$1$ element from~$S$,~$1$ element from~$S'$,~$8$ elements from~$S$,~$8$ elements from~$S'$,~$3$ elements from~$S$, \ldots, and finally~$2$ elements from~$S$.
Both runs~$S$ and~$S'$ have length~$11m$.
\label{fig:TS-superlinear}}
\end{figure}

Finally, our array~$A$ is built as follows.
First, we choose a base length~$\ell \geqslant \max\{8\btinit+19,\ms\}$ divisible by~$11$, and we set~$k = \lfloor \log_2(\ell) \rfloor$.
The array will be a permutation of~$\{1,2,\ldots,2^k \ell\}$, initially subdivided in~$2^k$ runs of length~$\ell$.
This ensures that the merge tree induced on~$A$ by each of the algorithms mentioned in Section~\ref{sec:description} is perfectly balanced: each run at height~$h$ in the tree has length~$2^h \ell$, and~$A$ has length~$n = 2^k \ell$.

Once this merge tree is fixed, we give a value between~$1$ and~$2^h\ell$ to each element of~$A$, by using the following top-down procedure.
When a run~$\overline{S}$ of length~$2^{h+1} \ell$ results from merging two runs~$S$ and~$S'$ of length~$2^h \ell$, we assign values of~$\overline{S}$ to either~$S$ or~$S'$ according to our second base brick:
the first value of~$\overline{S}$ comes from~$S'$, the second one from~$S$, the third one from~$S'$, the next~$8$ values come from~$S$, and so on.
The only exception to this rule is the first merge:
here, we assign according to our first base brick the values of the run~$\overline{R}$ to the two runs~$R$ and~$R'$ it results from.

Doing so ensures that using \TS's merging routine will perform at least~$0$ comparison for its first merge, and~$23 \times 2^h \ell / 11 - 3$ comparisons for each subsequent merge between runs at height~$h$.
There are~$2^{k-1}-1$ such merges for~$h = 0$, and~$2^{k-h-1}$ such merges for each height~$h$ such that~$1 \leqslant h \leqslant k-1$.
Hence, a total of
\begin{align*}
\mathsf{V}_A
& \geqslant (2^{k-1}-1)\left(\frac{23 \ell}{11}-3\right) + \sum_{h=1}^{k-1} 2^{k-h-1}\left(\frac{23 \times 2^h \ell}{11} - 3\right) \\
&= \frac{23 k n}{22} - \frac{23 \ell}{11} - 3 \times 2^k + 6
\end{align*}
comparisons are performed.
By contrast, using a naïve merging strategy would lead to performing only~$\mathsf{U}_A = k n - 2^k + 1$ comparisons.

Finally, the inequalities~$\ell^2 \geqslant 2^k \ell = n \geqslant 2^{2k} \geqslant 4 \ell^2$ prove that both~$\ell$ and~$2^k$ are~$\Theta(\sqrt{n})$, whereas~$k \sim \log_2(n)/2$.
Hence,~$\mathsf{U}_A \sim n \log_2(n) / 2$ and~$\mathsf{V}_A \geqslant 23 n \log_2(n) / 44 + o(n \log(n))$.
\end{proof}

\subsection{Sorting arrays with three values in super-linear time}
\label{sec:TS-routine-n-log-n}

\TS's merging routine was invented precisely with the goal of decreasing the number of comparisons performed when sorting arrays with few values.
In particular, sorting arrays of length~$n$ with a constant number of values should require only~$\mathcal{O}(n)$ comparisons.
Like in Section~\ref{sec:TS-routine-overhead}, we prove below that this is not the case.

\begin{proposition}
\label{pro:super-linear-sigma=5}
Let~$\A$ be one of the algorithms studied in Section~\ref{sec:pos-fast-growth}.
If~$\A$ uses \TS's routine for merging runs, it may require up to~$\Omega(n \log(n))$ element comparisons to sort arrays of length~$n$ with~$\sigma = 3$ distinct values.
\end{proposition}

\begin{proof}
Like for Proposition~\ref{pro:super-linear-overhead}, we explicitly build the array~$A$ by using building blocks that we assemble according to the description of Figure~\ref{fig:TS-worst-case}.
Below, we fix an arbitrarily large parameter~$p$, and we set~$\sL_p = (4\ms+1) (2^p + p-1+\btinit)+12$.
We also set~$\sR(x) = 2^{\lfloor \log_2(x) \rfloor}$ for all~$x \geqslant 1$, and~$t_k = k+s_2(k) + \btinit$ for all~$k \geqslant 1$, where~$s_2(k)$ is the sum of the binary digits of~$k$.

\begin{figure}[b!]
\begin{center}
\begin{tikzpicture}[scale=0.32]
\begin{scope}[shift={(3,0)}]
\foreach \i in {1,...,4}{
 \draw(5.5*\i,0) --++ (0,2);
}
\foreach \i in {0,...,2}{
 \node at (5.5*\i+5.5*0.5,1) {$\bX(\i)$};
}
\node at (5.5*3+5.5*0.5+0.1,1) {$\cdots$};
\node at (5.5*4+5.5*0.5,1) {$\bX(2^p-1)$};

\draw[very thick] (0,0) rectangle (27.5,2);
\node[anchor=north] at (13.75,-0.1) {General array structure};

\draw[very thick,<->,>=stealth] (0,2.7) --++ (27.5,0);
\node[anchor=south] at (13.75,2.7) {$2^p$ blocks};
\end{scope}

\begin{scope}[shift={(0,-9)}]
\foreach \i in {2,4,6,8,10,19,25,34,43}{
 \draw(\i,0) --++ (0,2);
}
\foreach \x/\l in {1/\bU,3/\bU,5/\bU,7.1/\cdots,9/\bU,%
14.5/00\cdots0,22/11\cdots1,29.5/22\cdots2,%
38.5/00\cdots0,46/11\cdots1}{
 \node at (\x,1){$\l$};
}

\foreach \y/\list in {-0.7/{0/10,10/24,34/15},%
2.7/{0/10,10/9,19/6,25/9,34/9,43/6},%
4.5/{0/34,34/15}}{
 \foreach \x/\l in \list{
  \draw[very thick,<->,>=stealth] (\x,\y) --++ (\l,0);
}}
\node[anchor=south] at (5,2.7) {$\sR\smash{(t_k)}$ runs};
\foreach \x in {22,46}{
 \node[anchor=south] at (\x,2.7) {\emph{\smash{p}addin\smash{g}}};
}
\foreach \x/\l in {14.5/6,29.5/12,38.5/6}{
 \node[anchor=south] at (\x,2.7) {$t\smash{_k}\!+\!\l$ elements};
}
\foreach \x in {17,42}{
 \node[anchor=south] at (\x,4.6) {$\sL\smash{_p}$ elements};
}

\draw[very thick] (0,0) rectangle (49,2);
\node[anchor=north] at (24.5,-2.7) {Block~$\bX(k)$};

\foreach [count=\k] \x in {5,22,41.5}{
 \node[anchor=north] at (\x,-0.7) {Part~$\bX_\k(k)$};
}
\end{scope}

\begin{scope}[shift={(36,0)}]
\foreach \i in {2,8}{
 \draw(\i,0) --++ (0,2);
}
\foreach \x/\l in {1/0,5/11\cdots1,9/2}{
 \node at (\x,1){$\l$};
}

\draw[very thick,<->,>=stealth] (2,2.7) --++ (6,0);
\draw[very thick,<->,>=stealth] (0,4.5) --++ (10,0);
\node[anchor=south] at (5,2.7) {\emph{\smash{p}addin\smash{g}}};
\node[anchor=south] at (5,4.6) {$\ms$ elements};

\draw[very thick] (0,0) rectangle (10,2);
\node[anchor=north] at (5,-0.1) {Run~$\bU$};
\end{scope}
\end{tikzpicture}
\vspace{-1.7em}
\end{center}
\caption{Bad array for \TS's galloping update policy, containing only values~$0$,~$1$ and~$2$.\label{fig:TS-worst-case}}
\end{figure}

The array~$A$ that we build is divided in~$2^p$ blocks~$\bX(k)$ of length~$2\sL_p$.
Each block~$\bX(k)$ is subdivided into three parts~$\bX_1(k)$,~$\bX_2(k)$ and~$\bX_3(k)$, whose lengths~$x_1(k)$,~$x_2(k)$ and~$x_3(k)$ obey the relations~$x_1(k) = \ms \, \sR(t_k)$ and~$x_1(k) + x_2(k) = x_3(k) = \sL_p$.
The part~$\bX_1(k)$ itself is subdivided into~$\sR(t_k)$ runs~$\bU$ of length~$\ms$; each of the parts~$\bX_2(k)$ and~$\bX_3(k)$ consists of just one run.

By construction,~$t_k \leqslant 2^p-1 + s_2(2^p-1) + \btinit = 2^p + (p-1) + \btinit$ whenever~$0 \leqslant k \leqslant 2^p-1$, and~$x+1 \leqslant \sR(x) \leqslant 2x$ for all integers~$x \geqslant 1$.
It follows that
\[x_2(k)-x_1(k) = \sL_p - 2 \ms \, \sR(t_k) \geqslant (4\ms+1) t_k + 12 - 4 \ms \, t_k = t_k + 12 \geqslant 0.\]
Hence, the algorithm~$\A$ proceeds in three phases, interleaved with each other:
(i)~it sorts each part~$\bX_1(k)$ by recursively merging the runs~$\bU$ of which~$\bX_1(k)$ consists, following a perfectly balanced binary merge tree with~$\sR(t_k)$ leaves --- let~$\bX_1(k)$ abusively denote the resulting run;
(ii)~it merges~$\bX_1(k)$ with~$\bX_2(k)$, and merges the resulting run (say, ~$\bX_{1+2}(k)$) with~$\bX_3(k)$;
(iii)~it merges the sorted blocks~$\bX(k)$ with each other, again following a  perfectly balanced binary merge tree with~$2^p$ leaves.
Moreover, merges are performed according to a post-order traversal of the merge tree.

For instance, the merge tree obtained when~$p = 2$ and~$\btinit = 7$ is presented in Figure~\ref{fig:TS-worst-case-tree}; its inner nodes are labelled chronologically: the node labelled~$k$ results from the~$k$\textsuperscript{th} merge that~$\A$ performs.	

\begin{figure}[h!]
\begin{center}
\begin{tikzpicture}[scale=0.4]
\draw[black!50,ultra thick,pattern=north west lines,pattern color=black!50]
(18,6.7) -- (38.2,4.86364) -- (38.2,3.7)%
--++ (-0.4,-0.4) --++ (-0.6,0) --++ (-0.4,0.4) --++ (-9.6,0)%
--++ (-0.4,-0.4) --++ (-0.6,0) --++ (-0.4,0.4) --++ (-9.6,0)%
--++ (-0.4,-0.4) --++ (-0.6,0) --++ (-0.4,0.4) --++ (-9.6,0)%
--++ (-0.4,-0.4) --++ (-0.6,0) --++ (-0.4,0.4) -- (3.8,5.5) -- (17,6.7) -- cycle;

\draw[black!50,ultra thick,fill=black!10] (-0.7,-0.7) --++ (0,1.7) --++ (1.7,1.7) --++ (1,0) --++ (1.7,-1.7) --++ (0,-1.7) -- cycle;
\foreach \x in {7,18,29}{
 \draw[black!50,ultra thick,fill=black!10] (\x-0.7,-1.7) --++ (0,2.55) --++ (3.7,1.85) --++ (1,0) --++ (3.7,-1.85) --++ (0,-2.55) -- cycle;
}

\draw[thick] (1.5,2) --++ (2,1) ++ (1,1) --++ (2,1) --++ (9,-1)
(6.5,5) -- (17.5,6) -- (28.5,5)
(26.5,4) --++ (2,1) --++ (9,-1);
\foreach \x in {7,18,29}{
 \draw[thick] (\x+1.5,1) --++ (2,1) --++ (2,-1) ++ (-2,1) --++ (4,1);
 \draw[fill=white,thick] (\x+3.5,2) circle (0.43);
}
\foreach \x/\y in {0/0,7/-1,11/-1,18/-1,22/-1,29/-1,33/-1}{
 \draw[thick] (\x,\y) --++ (0.5,1) ++ (1.5,-1) --++ (0.5,1) --++ (0.5,-1);
 \foreach \z in {0,2,3}{
  \draw[fill=white,thick] (\x+\z,\y) circle (0.43);
 }
}
\foreach \x/\y in {0/0,3/2,7/-1,11/-1,14/2,18/-1,22/-1,25/2,29/-1,33/-1,36/2}{
 \draw[thick] (\x+1,\y) --++ (-0.5,1) --++ (1,1) --++ (1,-1);
 \foreach \u/\v in {1/0,0.5/1,2.5/1,1.5/2}{
  \draw[fill=white,thick] (\x+\u,\y+\v) circle (0.43);
}}
\foreach \x/\y in {6/5,28/5,17/6}{
 \draw[fill=white,thick] (\x.5,\y) circle (0.43);
}
\foreach [count=\k] \x/\y in {0/1,2/1,1/2,3/3,4/4,%
7/0,9/0,8/1,11/0,13/0,12/1,10/2,14/3,15/4,6/5,%
18/0,20/0,19/1,22/0,24/0,23/1,21/2,25/3,26/4,%
29/0,31/0,30/1,33/0,35/0,34/1,32/2,36/3,37/4,28/5,17/6}{
 \node at (\x.5,\y) {\tiny \k};
}

\foreach \x/\y/\k in {1/0/0,10/-1/1,21/-1/2,32/-1/3}{
 \node[anchor=north] at (\x.5,\y-0.7) {\color{black!60}$\bX_1(\k)$};
}
\foreach \k in {2,3}{
 \draw[thick,->,>=stealth] (3,-1.2*\k-0.5) -- (1+1.5*\k,-1.2*\k-0.5) -- (1+1.5*\k,\k-0.55);
 \node[anchor=east] at (3.07,-1.2*\k-0.5) {$\bX_{\k}(0)$};
}
\foreach \i in {1,2,3}{
 \foreach \k in {2,3}{
  \draw[thick,->,>=stealth] (11*\i+1,-1.2*\k-1.5) -- (11*\i+1+1.5*\k,-1.2*\k-1.5) -- (11*\i+1+1.5*\k,\k-0.55);
  \node[anchor=east] at (11*\i+1.07,-1.2*\k-1.5) {$\bX_{\k}(\i)$};
}
}
\end{tikzpicture}
\vspace{-1.7em}
\end{center}
\caption{Merge tree induced by the algorithm~$\A$ on the array~$A$.
The four balanced gray sub-trees result in sorting the parts~$\bX_1(k)$;
the balanced hatched sub-tree results in sorting~$A$ itself once each block~$\bX(k)$ is sorted.
Each inner node is labelled~$k$ if it results from the~$k$\textsuperscript{th} merge that~$\A$ performs.\label{fig:TS-worst-case-tree}}
\end{figure}

Once the order in which~$\A$ performs merges is known, Algorithm~\ref{alg:TS-merge} allows us to track the dynamics of the parameter~$\bt$.
For the ease of the explanation, let~$\bX(2^h,\ell)$ denote the run obtained by merging the blocks~$\bX(2^h \ell), \bX(2^h \ell+1),\ldots, \bX(2^h (\ell+1)-1)$.
The key invariant of our construction is two-fold:
\begin{enumerate}
\item the parameter~$\bt$ is equal to the integer~$t_k = k + s_2(k) + \btinit$ just before~$\A$ starts sorting a block~$\bX(k)$, and it is equal to~$t_k+2$ just after~$\bX(k)$ has been sorted;
\item the parameter~$\bt$ is equal to the integer~$u_{h,\ell} = t_{2^h \ell} + 2^{h+1} + 2$ just before~$\A$ starts merging two runs~$\bX(2^h,\ell)$ and~$\bX(2^h,\ell+1)$, and it is equal to~$u_{h,\ell}-1$ just after these runs have been merged.
\end{enumerate}

To prove this invariant, we first show that, if~$\bt = t_k$ when~$\A$ starts sorting~$\bX(k)$, Algorithm~\ref{alg:TS-merge} will keep calling the function \textsf{LRMerge} in line~\ref{alg:TS:LR}.
Then, the only time at which~$\bt$ may be changed is just after discovering~$b_0$ by using~$\bt$-galloping:~$\bt$ decreases if~$b_0 \geqslant \bt+9$, it increases if~$\bt+8 \geqslant b_0 \geqslant \bt+2$, and does not vary if~$\bt+1 \geqslant b_0$.

\begin{figure}[t!]
\begin{center}
\begin{tikzpicture}[scale=0.32]
\foreach \x in {0,1,2}{
 \foreach \y in {2,4}{
  \draw[very thick] (17*\x,\y) --++ (7.5,0) ++ (0.5,0) --++ (7.5,0);
 }
 \draw[very thick] (17*\x,0) rectangle++ (7.5,6) (17*\x+8,0) rectangle++ (7.5,6);
}
\foreach \x/\list/\t in {3.75/{2\smash{^h},?,2\smash{^h}}/\phantom{(}2\smash{^h} \bU\phantom{(},%
11.75/{2\smash{^h},?,2\smash{^h}}/\phantom{(}2\smash{^h} \bU\phantom{(},%
20.75/{\sR\smash{(t_k)},?,\sR\smash{(t_k)}}/\bX_1(k),%
28.75/{t\smash{_k}\!+\!6,?,t\smash{_k}\!+\!12}/\bX_2(k),%
37.75/{\sR\smash{(t_k)}\!+\!t\smash{_k}\!+\!6,?,\sR\smash{(t_k)}\!+\!t\smash{_k}\!+\!12}/\bX_{1+2}(k),%
45.75/{t\smash{_k}\!+\!6,?,0}/\bX_3(k)}{
 \foreach [count=\y] \l in \list{
 \node at (\x,2*\y-1) {$\l$};
}
 \node[anchor=north] at (\x,0) {$\t$};
}
\draw[very thick,densely dashed] (16.25,-1.5) --++ (0,8);
\draw[very thick,densely dashed] (33.25,-1.5) --++ (0,8);
\end{tikzpicture}
\vspace{-1.2em}
\end{center}
\caption{Runs merged while sorting~$\bX(k)$.
Each run is represented by a~$3$-row array: the integer in the lowest (resp., highest) cell is the number of elements with value~$0$ (resp.,~$2$) of the run.
Elements with value~$1$ serve as padding runs: counting them exactly is irrelevant.
From left to right, we see the merges between (i)~two runs that each result from merging~$2^h$ runs~$\bU$; (ii)~the sorted part~$\bX_1(k)$ with~$\bX_2(k)$; (iii)~$\bX_{1+2}(k)$ and~$\bX_3(k)$.
\label{fig:TS-worst-case-details}}
\end{figure}

Now, let us consider each merge performed to sort~$\bX(k)$;
the composition of the intermediate runs obtained while doing so is presented detail in Figure~\ref{fig:TS-worst-case-details}.
First, while merging sorted clusters of~$2^h$ contiguous runs~$\bU$, we have~$a - a_2 = b - b_2 = (\ms-1) 2^h$ and~$b_0 = 2^h \leqslant \sR(t_k)/2 \leqslant t_k = \bt$.
Then, when merging~$\bX_1(k)$ with~$\bX_2(k)$, we have~$a - a_0 \leqslant a = x_1(k) \leqslant x_2(k) - (t_k+12) = b - b_2$ and~$\bt+8 \geqslant b_0 = \bt+6 \geqslant \bt+2$, which leads to increasing~$\bt$.
Finally, when merging~$\bX_{1+2}(k)$ with~$\bX_3(k)$, we have~$a - a_0 = \sL_p - a_0 \leqslant \sL_p = b - b_2$ and~$\bt+8 \geqslant b_0 = \bt+5 \geqslant \bt+2$, which leads to increasing~$\bt$ once more.
It follows, as promised, that~$\bt = t_k+2$ just after~$\A$ has sorted~$\bX(k)$.

The second step towards proving the invariant consists in showing that, if~$\bt = u_{h,\ell}$ when~$\A$ starts merging two runs~$\bX(2^h,\ell)$ and~$\bX(2^h,\ell+1)$, it decreases~$\bt$ once.
In other words, we shall prove that either Algorithm~\ref{alg:TS-merge} calls the function \textsf{LRMerge} in line~\ref{alg:TS:LR} and~$b_0 \geqslant u_{h,\ell}+9$, or it calls \textsf{RLMerge} in line~\ref{alg:TS:RL} and~$a_2 \geqslant u_{h,\ell}+9$.
In practice, we will simply prove that~$\min\{b_0, a_2\} \geqslant u_{h,\ell+9}$.

Indeed, observe that~$\ell$ is even, and thus that~$t_{2^h \ell + i} = t_{2^h \ell} + i + s_2(i)$ whenever~$0 \leqslant i < 2^{h+1}$.
Thus, when~$2^h \ell \leqslant k < 2^h(\ell+2)$, the run~$\bX(k)$ contains~$\sR(t_k)+2t_k+12 \geqslant 3 t_k+13 \geqslant 2 t_{2^h \ell} + 13$ elements with value~$0$ and~$\sR(t_k)+t_k+12 \geqslant 2 t_k+13 \geqslant 2 t_{2^h \ell} + 13$ elements with value~$2$.
It follows that~$\min\{a_2,b_0\} \geqslant 2^h (t_{2^h \ell} + 13) \geqslant (t_{2^h \ell} + 11) + 2^h \times 2 = u_{h,\ell}+9$.

We prove now our invariant by induction on the number of merges performed by the algorithm~$\A$: we distinguish five types of merges:
\begin{itemize}
\item When sorting the block~$\bX(0)$, the algorithm starts with a parameter~$\bt = \btinit$.

\item When sorting a block~$\bX(k)$, where~$k$ is odd, it has just finished merging the block~$\bX(k-1)$, leaving us with a parameter~$\bt = t_{k-1}+2 = t_k$.

\item When sorting a block~$\bX(2^h k)$, where~$k$ is odd and~$h \geqslant 1$, it has just finished merging the runs~$\bX(2^{h-1},2k-2)$ and 
$\bX(2^{h-1},2k-1)$, leaving us with a parameter
\[\bt = u_{h-1,2k-2}-1 = t_{2^h (k-1)} + 2^h + 1 = 2^h (k-1) + s_2(k-1) + 2^h + 1 = t_{2^h k}.\]

\item When merging two runs~$\bX(2^h,k)$ and~$\bX(2^h,k+1)$, where~$k$ is even and~$h = 0$, it has just finished sorting the block~$\bX(k+1)$, leaving us with a parameter~$\bt = t_{k+1}+2 = t_k+4 = u_{0,k}$.

\item When merging two runs~$\bX(2^h,k)$ and~$\bX(2^h,k+1)$, where~$k$ is even and~$h \geqslant 1$, it has just finished merging the runs~$\bX(2^{h-1},2k+2)$ and~$\bX(2^{h-1},2k+3)$, leaving us with a parameter
\[\bt = u_{h-1,2k+2}-1 = t_{2^h(k+1)} + 2^h + 1 = 2^h(k+1) + s_2(k+1) + 2^h + 1 = u_{h,k}.\]
\end{itemize}

Equipped with this invariant, we can finally compute a lower bound on the number of comparisons that~$\A$ performs.
More precisely, we will count only comparisons performed naïvely while sorting parts~$\bX_1(k)$.
When sorting such a part, we recursively merge sorted clusters of~$2^h$ runs~$\bU$.
To do so, we call the \textsf{LRMerge} function, and naïvely discover~$b_0 = 2^h$ elements with value~$0$ in the right run we are merging.
Thus, in total, we use at least~$\sR(t_k) \log_2(\sR(t_k))$ such comparisons to sort~$\bX_1(k)$, and at least
\[\mathsf{C}_p = \sum_{k = 2^{p-1}}^{2^p-1} \sR(t_k) \log_2(\sR(t_k)) 
\geqslant \sum_{k = 2^{p-1}}^{2^p-1} \sR(k) \log_2(\sR(k)) \geqslant 2^{p-1} \sR(2^{p-1}) \log_2(\sR(2^{p-1})) = 2^{2p-1}p\]
comparisons to sort the array~$A$ itself.
Observing that~$A$ is of length~$n = 2^{p+1} \sL_p \sim (4\ms+1)2^{2p+1}$ proves that~$n \log_2(n) \sim (4\ms+1)2^{2p+2}p \in \mathcal{O}(\mathsf{C}_p)$, which completes the proof.
\end{proof}

In spite of this negative result, we can still prove that \TS's routine and update strategy is harmless when sorting arrays with only two values, thereby making our above construction somehow \emph{as simple as possible}.

\begin{proposition}
\label{pro:linear-sigma=2}
Let~$\A$ be a stable natural merge sort algorithm with the middle-growth property.
If~$\A$ uses \TS's actual routine (including the heuristics for updating the parameter~$\bt$), it requires~$\O(n)$ element comparisons to sort arrays of length~$n$ with two values.
\end{proposition}

\begin{proof}
When merging two runs containing only~$\sigma = 2$ values, Algorithm~\ref{alg:TS-merge} just uses twice~$0$-galloping and once~$(\bt+1)$-galloping, and stops without updating~$\bt$.
Thus, up to wasting a maximum of~$\bt+1$ comparisons per merge,~$\A$ keeps using~$0$-galloping, and Theorem~\ref{thm:middle-few} proves that doing so requires only~$\O(n)$ comparisons in total to sort arrays of length~$n$ with~$2$ values.
\end{proof}

\section{Refined complexity bounds}
\label{sec:precise-bounds-PoS}

One weakness of Theorem~\ref{thm:middle-few} is that it cannot help us to distinguish the complexity upper bounds of those algorithms that have the middle-growth property, although the constants hidden in the~$\O$ symbol could be dramatically different.
Below, we study these constants, and focus on upper bounds of the type~$\mathsf{c} n \H^\ast + \O(n)$ or~$\mathsf{c} n (1 + o(1)) \H^\ast + \O(n)$.

Since sorting arrays of length~$n$, in general, requires at least~$\log_2(n!) = n \log_2(n) + \O(n)$ comparisons, and since~$\H^\ast \leqslant \log_2(n)$ for all arrays, we already know that~$\mathsf{c} \geqslant 1$ for any such constant~$\mathsf{c}$.
Below, we focus on finding matching upper bounds in two regimes: first using a fixed parameter~$\bt$, thereby obtaining a constant~$\mathsf{c} > 1$, and then letting~$\bt$ depend on the lengths of those runs that are being merged, in which case we reach the constant~$\mathsf{c} = 1$.

Inspired by the success of Theorem~\ref{thm:middle-few-naïve}, which states that algorithms with the tight middle-growth property sort array of length~$n$ by using only~$\mathsf{c} n \log_2(n) + \O(n)$ element comparisons with~$\mathsf{c} = 1$, we focus primarily on that property, while not forgetting other algorithms that also enjoyed~$n \H + \O(n)$ or~$n \log_2(n) + \O(n)$ complexity upper bounds despite not having the tight middle-growth property.

\subsection{Fixed parameter}
\label{subsec:fixed-t}

\begin{lemma}
\label{lem:slice-1}
Let~$\T$ be a merge tree induced a stable algorithm on some array~$A$ of length~$n$ with~$\sigma$ dual runs~$S_1,S_2,\ldots,S_\sigma$.
Consider a fixed parameter~$\bt \geqslant 0$, a real number~$u > 1$, and some index~$i \leqslant \sigma$.
Then, for all~$h \geqslant 0$, let~$\T_h$ be a set of pairwise incomparable nodes of~$\T$ (i.e., no node of~$\T_h$ descends from an other one) such that each run~$R$ in~$\T_h$ is of length~$r \geqslant u^h (\bt+1) n / s_i$.
We have
\[\sum_{R \in \T_{\geqslant 0}} \cost_{\bt}^\ast(r_{\rightarrow i}) \leqslant \frac{9 u}{u-1}s_i,\]
where~$\T_{\geqslant 0}$ denotes the union of the sets~$\T_h$.
\end{lemma}

\begin{proof}
For each integer~$h \geqslant 0$, let~$\C(h) = \sum_{R \in \T_h} \cost_{\bt}^\ast(r_{\rightarrow i})$.
Let also~$f \colon x \mapsto \bt + 2 + 2 \log_2(x+1)$ and~$g \colon x \mapsto x \, f(s_i / x)$ .
Both functions~$f$ and~$g$ are concave and increasing on~$(0,+\infty)$.

Then, let~$v = \bt+1$.
Since~$\T_h$ consists of pairwise incompatible runs, we have~$\sum_{R \in \T_h} r_{\rightarrow i} \leqslant s_i$ and~$u^h v n |\T_h| / s_i \leqslant \sum_{R \in \T_h} \leqslant n$, i.e.,~$|\T_h| \leqslant s_i / (u^h v)$.
It follows that
\begin{align*}
\C(h)
& \leqslant \sum_{R \in \T_h} f(r_{\rightarrow i})
\leqslant |\T_h| f\big(\!\sum_{R \in \T_h}\!r_{\rightarrow i} / |\T_h|\big)
\leqslant g(|\T_h|) \\
& \leqslant g\Big(\frac{s_i}{u^h v}\Big) = \frac{s_i}{u^h v} f(u^h v) \\ & \leqslant \frac{s_i}{u^h v}(\bt+2+2\log_2(2^{v+1} u^h)) = \frac{s_i}{u^h v}(3v+3+2 h \log_2(u)) \\
& \leqslant \frac{s_i}{u^h}(6+2 h \log_2(u)).
\end{align*}
Denoting the latter expression by~$\C_+(h)$ and observing that~$\log_2(u) \leqslant 3 (u-1)/2$ for all~$u > 1$, we conclude that
\[\sum_{R \in \T_{\geqslant 0}} \cost_{\bt}^\ast(r_{\rightarrow i}) = \sum_{h \geqslant 0} \C(h) \leqslant \sum_{h \geqslant 0} \C_+(h) = \frac{u s_i}{u-1}\Big(6+2 \frac{\log_2(u)}{u-1}\Big) \leqslant \frac{9 u}{u-1}s_i.\qedhere\]
\end{proof}

\begin{theorem}\label{thm:PoS-constant}
Let~$\A$ be a stable natural merge sort algorithm with the tight middle-growth property.
For each parameter~$\bt \geqslant 0$, if~$\A$ uses the~$\bt$-galloping routine for merging runs, it requires at most
\[(1 + 1 / (\bt+3)) n \H^\ast + \log_2(\bt+1) n + \O(n)\]
element comparisons to sort arrays of length~$n$ and dual run-length entropy~$\H^\ast$.
\end{theorem}

\begin{proof}
Let us follow a variant of the proof of Theorem~\ref{thm:middle-few}.
Let~$\gamma$ be the integer mentioned in the definition of the statement ``$\A$ has the tight middle-growth property'', let~$\T$ be the merge tree induced by~$\A$ on an array~$A$ of length~$n$, and let~$s_1,s_2,\ldots,s_\sigma$ be the lengths of the dual runs of~$A$.
Like in the proof of Theorem~\ref{thm:middle-few}, we just need to prove that
\[\sum_{R \in \T} \cost_{\bt}^\ast(r_{\rightarrow i}) \leqslant
(1+1/(\bt+3)) \log_2(n / s_i) s_i + \log_2(\bt+1) s_i + \O(s_i)\]
for all~$i \leqslant \sigma$.

Let~$\R_h$ be the set of runs at height~$h$ in~$\T$.
By construction, no run in~$\R_h$ descends from another one, which, like in Lemma~\ref{lem:slice-1}, proves that~$\sum_{R \in \R_h} r_{\rightarrow i} \leqslant s_i$.
Thus, if we set
\[\C_{\bt}(h) = \sum_{R \in\R_h} \cost_{\bt}^\ast(r_{\rightarrow i}),\]
it follows that
\[\C_{\bt}(h) \leqslant
(1+1/(\bt+3)) \sum_{R \in \R_h} r_{\rightarrow i} \leqslant (1+1/(\bt+3)) s_i\]
for all~$h \geqslant 0$.

Now, let~$\mu = \lceil\log_2((\bt+1) n / s_i)\rceil+\gamma$, and let~$\T_h = \R_{h+\mu}$.
By construction, each run~$R$ belonging to the set~$\T_h = \R_{h + \mu}$ is of length~$r \geqslant 2^{h + \mu-\gamma} \geqslant 2^h(\bt+1)n/s_i$.
Thus, applying Lemma~\ref{lem:slice-1} to~$u = 2$ indicates that~$\sum_{h \geqslant \mu} \C_{\bt}(h) \leqslant 18 s_i = \O(s_i)$, and we conclude as desired that
\[\sum_{h \geqslant 0} \C_{\bt}(h) \leqslant (1+1/(\bt+3)) \mu s_i + \O(s_i)
= (1+1/(\bt+3)) \log_2(n/s_i) s_i + \log_2(\bt+1) s_i + \O(s_i). \qedhere\]
\end{proof}

\subsection{Polylogarithmic parameter}
\label{subsec:variable-t}

Letting the parameter~$\bt$ vary, we minimise the upper bound provided by Theorem~\ref{thm:PoS-constant} by choosing~$\bt = \Theta(\H^\ast+1)$, in which case this upper bound simply becomes~$n \H^\ast + \log_2(\H^\ast+1)n + \O(n)$.
However, computing or approximating~$\H^\ast$ before starting the actual sorting process would be both rather unreasonable and not necessarily worth the effort;
indeed, Theorem~\ref{thm:PoS-log} would make this preliminary step useless unless it requires less than~$\log_2(\H^\ast+1)n$ element comparisons, which seems extremely difficult.
Instead, we update the parameter~$\bt$ as follows, which will provide us with a slightly larger upper bound.

\begin{definition}
\label{def:log-galloping}
We call \emph{polylogarithmic} galloping routine the merging routine that, when merging adjacent runs of lengths~$a$ and~$b$, performs the same comparisons and element moves as the~$\bt$-galloping routine for~$\bt = \lceil \log_2(a+b) \rceil^2$.
\end{definition}

We first prove that the overhead of using galloping with this update strategy instead of using a naïve merging routine is at most linear.

\begin{lemma}\label{lem:poly-slice-2}
Let~$\A$ be a stable algorithm with the middle-growth property.
Let~$\T$ be a merge tree induced by~$\A$ on some array~$A$ of length~$n$ with~$\sigma$ dual runs~$S_1,S_2,\ldots,S_\sigma$, and let~$\T^\ast$ be the set of internal nodes of~$\T$.
If~$\A$ uses the polylogarithmic routine for merging runs, it requires no more than
\[\sum_{R \in \T^\ast} \sum_{i = 1}^\sigma \cost^\ast_{\log}(r,r_{\rightarrow i}) +\O(n)\]
comparisons to sort~$\A$, where we set~$\cost^\ast_{\log}(r,m) = \min\{m,6 \log_2(r+1)^2+6\}$.
\end{lemma}

\begin{proof}
Using a parameter~$\bt = \lceil \log_2(r) \rceil^2$ to merge runs~$R'$ and~$R''$ into one run~$R$ requires at most
\[1 + \sum_{i=1}^\sigma
\cost_{\lceil \log_2(r) \rceil^2}^\ast(r'_{\rightarrow i}) +
\cost_{\lceil \log_2(r) \rceil^2}^\ast(r''_{\rightarrow i})\]
element comparisons.
Given that
\begin{align*}
\cost_{\lceil \log_2(r) \rceil^2}^\ast(r'_{\rightarrow i})
& \leqslant \min\{(1+1/\log_2(r)^2) r'_{\rightarrow i}, \log_2(r)^2+3+2 \log_2(r'_{\rightarrow i}+1)\} \\
& \leqslant \min\{r'_{\rightarrow i}, 3\log_2(r+1)^2+3\} + r'_{\rightarrow i} / \log_2(r)^2
\end{align*}
and that~$r'_{\rightarrow i} + r''_{\rightarrow i}=r_{\rightarrow i}$, this makes a total of at most
\[1 + \frac{r}{\log_2(r)^2} + \sum_{i=1}^\sigma \cost^\ast_{\log}(r,r_{\rightarrow i})\]
element comparisons.

Now, let~$\beta > 1$ be the real number mentioned in the definition of the statement ``$\A$ has the middle-growth property'', and let~$\R_h$ be the set of runs at height~$h$ in~$\T^\ast$.
The lengths of runs in~$\R_h$ sum up to~$n$ or less, and thus
\[\sum_{R \in \T^\ast} \frac{r}{\log_2(r)^2} = \sum_{h \geqslant 1} \sum_{R \in \R_h} \frac{r}{\log_2(r)^2} \leqslant \sum_{h \geqslant 1} \sum_{R \in \R_h} \frac{r}{h^2 \log_2(\beta)^2} \leqslant \sum_{h \geqslant 1} \frac{n}{h^2 \log_2(\beta)^2} = \O(n).\]
We conclude by remembering that~$\A$ makes~$n-1$ comparisons to identify runs prior to merging them and performs~$\rho-1 \leqslant n-1$ merges.
\end{proof}

\begin{proposition}
\label{pro:linear-overhead}
Let~$\A$ be a stable natural merge sort algorithm with the middle-growth property.
Let~$\mathsf{U}_A$ be the number of comparisons performed by~$\A$ to sort an array~$A$ when using a naïve merging routine, and let~$\mathsf{V}_A$ be the number of comparisons performed by~$\A$ when using the polylogarithmic galloping routine.
For all arrays~$A$ of length~$n$, we have~$\mathsf{V}_A \leqslant \mathsf{U}_A + \O(n)$.
\end{proposition}

\begin{proof}
Lemma~\ref{lem:poly-slice-2} proves that
\[\mathsf{V}_A = \sum_{R \in \T^\ast} \sum_{i=1}^\sigma \cost^\ast_{\log}(r,r_{\rightarrow i}) + \O(n) \leqslant \sum_{R \in \T^\ast} \sum_{i=1}^\sigma r_{\rightarrow i} + \O(n) = \sum_{R \in \T^\ast} r + \O(n) = \mathsf{U}_A + \O(n). \qedhere\]
\end{proof}

In addition, the polylogarithmic galloping also turns out to be extremely efficient for algorithms with the (tight) middle-growth property.

\begin{lemma}
\label{lem:slice-2}
Let~$\T$ be a merge tree induced a stable algorithm on some array~$A$ of length~$n$ with~$\sigma$ dual runs~$S_1,S_2,\ldots,S_\sigma$.
Consider a real number~$u \in (1,2\,]$ and some index~$i \leqslant \sigma$.
Then, for all~$h \geqslant 0$, let~$\T_h$ be a set of pairwise incomparable nodes of~$\T$ such that each run~$R$ in~$\T_h$ is of length~$r \geqslant u^h \log_2(2 n / s_i)^2 n / s_i$.
We have
\[\sum_{R \in \T_{\geqslant 0}} \cost_{\log}^\ast(r,r_{\rightarrow i}) \leqslant
\frac{6 (10u^2-13u+5)u}{(u-1)^3} s_i,\]
where~$\T_{\geqslant 0}$ denotes the union of the sets~$\T_h$.
\end{lemma}

\begin{proof}
For each integer~$h \geqslant 0$, let~$\C(h) = \sum_{R \in \T_h} \cost_{\log}^\ast(r, r_{\rightarrow i})$.
Let~$f \colon x \mapsto 1 + \log_2(x+1)^2$ and~$g \colon x \mapsto x \, f(n / x)$; the function~$f$ is concave and increasing on~$[2,+\infty)$, and~$g$ is increasing on~$(0,+\infty)$.

Then, let~$z = n/s_i$ and~$v = \log_2(2z)^2$.
Since~$\T_h$ consists of pairwise incompatible runs, we have~$u^h v z |\T_h| \leqslant \sum_{R \in \T_h} r \leqslant n$, i.e.,~$|\T_h| \leqslant s_i / (u^h v)$.
It follows that
\begin{align*}
\C(h) / 6
& \leqslant \sum_{R \in \T_h} f(r)
\leqslant |\T_h| f\big(\!\sum_{R \in \T_h}\!r / |\T_h|\big)
\leqslant g(|\T_h|) \\
& \leqslant g(s_i / (u^h v)) = f(u^h v z) s_i / (u^h v) \\
& \leqslant (1 + \log_2(2^{h+3} z^3)^2) s_i / (u^h v) \\
& \leqslant (1 + (h+ 3\log_2(2z))^2) s_i / (u^h v) \\
& \leqslant  (h^2+6h+10)s_i / u^h.
\end{align*}
where the inequality between the second and third lines simply comes from the fact that
\[1 + u^h v z \leqslant 1 + 4 u^h z^3 \leqslant 8 u^h z^3 \leqslant 2^{h+3} z^3,\]
and the inequality between the last two lines comes from the fact that
\[1 + (h+ 3\log_2(2z))^2 \leqslant \log_2(2z)^2 + (h \log_2(2z) + 3 \log_2(2z))^2 = (h^2+6h+10) v.\]
Setting~$\C_+(h) = (h^2+6h+10) s_i / u^h$, we conclude that
\[\sum_{R \in \T_{\geqslant 0}} \cost_{\log}^\ast(r,r_{\rightarrow i}) \leqslant \sum_{h \geqslant 0} \C(h)
\leqslant 6 \sum_{h \geqslant 0} \C_+(h) = 6 \frac{(10u^2-13u+5)u}{(u-1)^3} s_i.\qedhere\]
\end{proof}

\begin{theorem}\label{thm:PoS-log}
Let~$\A$ be a stable natural merge sort algorithm with the middle-growth property.
If~$\A$ uses the polylogarithmic galloping routine for merging runs,
it requires~$\O(n + n\H^\ast)$element comparisons to sort arrays of length~$n$ and dual run-length entropy~$\H^\ast$.
If, furthermore,~$\A$ has the tight middle-growth property, it requires at most~$n \H^\ast + 2 \log_2(\H^\ast+1) n + \O(n)$ element comparisons to sort such arrays.
\end{theorem}

\begin{proof}
Let us refine and adapt the proofs of Theorems~\ref{thm:middle-few} and~\ref{thm:PoS-constant}.
Let~$\T$ be the merge tree induced by~$\A$ on an array~$A$ of length~$n$ with~$\sigma$ dual runs of lengths~$s_1,s_2,\ldots,s_\sigma$.
For all integers~$h \geqslant 0$, let~$\R_h$ be the set of runs at height~$h$ in~$\T$, and let~$\T^\ast$ be the set of internal nodes of~$\T$.
Let~$\beta \in (1,2]$ and~$\gamma \geqslant 0$ be a real number and an integer such that~$r \geqslant \beta^{h - \gamma}$ for all runs~$R$ of height~$h$ in~$\T$; if~$\A$ has the tight middle-growth property, we choose~$\beta = 2$.

Thanks to Lemma~\ref{lem:poly-slice-2}, we shall just prove that
\[\sum_{R \in \T^\ast} \cost^\ast_{\log}(r,r_{\rightarrow i}) \leqslant
\log_\beta(n/s_i) s_i + 2 \log_\beta(\log_2(n/s_i)+1) s_i + \O(s_i)\]
for all~$i \leqslant \sigma$.
Indeed, the algorithm~$\A$ will then perform no more than
\begin{align*}
\sum_{R \in \T^\ast}\sum_{i=1}^\sigma \cost^\ast_{\log}(r,r_{\rightarrow i}) + \O(n)
& \leqslant \sum_{i=1}^\sigma \log_\beta(n/s_i) s_i + 2 \log_\beta(\log_2(n/s_i)+1) s_i + \O(n) \\
& \leqslant \log_\beta(2) n \H^\ast + 2 \log_\beta(\H^\ast+1) n + \O(n)
\end{align*}
comparisons, the latter inequality being due to the concavity of the function~$x \mapsto \log_\beta(x+1)$.

Then, let~$\R_h$ be the set of runs at height~$h$ in~$\T$, and let
\[\C_{\log}(h) = \sum_{R \in \R_h} \cost^\ast_{\log}(r,r_{\rightarrow i}).\]
Once again,~$\C_{\log}(h) \leqslant \sum_{R \in \R_h} r_{\rightarrow i}
\leqslant \sum_{R \in \R_h} r_{\rightarrow i} = s_i$ for all~$h \geqslant 0$, because~$\R_h$ consists of pairwise incomparable runs~$R$.

Now, let~$\nu = \lceil\log_\beta(n/s_i \log_2(2n/s_i)^2)\rceil + \gamma$.
By construction, each run $R$ in~$\T_h = \R_{h + \nu}$ is of length
\[r \geqslant \beta^{h + \nu - \gamma} \geqslant \beta^h \log_2(2n/s_i)^2 n/s_i.\]
Thus, applying Lemma~\ref{lem:slice-2} to~$u = \beta$ indicates that~$\sum_{h \geqslant \nu} \C_{\bt}(h) = \O(s_i)$, and we conclude that
\[\sum_{h \geqslant 0} \C_{\bt}(h) \leqslant \nu s_i + \O(s_i)
= \log_\beta(n/s_i) s_i + 2 \log_\beta(\log_2(n/s_i)+1) s_i + \O(s_i). \qedhere\]
\end{proof}

Finally, in practice, we could improve our the~$n \H^\ast + 2\log_2(\H^\ast+1)n + \O(n)$ upper bound by adapting our update policy.
For instance, choosing~$\bt = \lceil \log_2(a+b) \rceil \times \lceil \log_2(\log_2(a+b))\rceil^2$ would slightly damage the constant hidden in the~$\O(n)$ side of the inequality~$\mathsf{V}_A \leqslant \mathsf{U}_A + \O(n)$, but would also reduce the number of comparisons required by~$\A$ to
\[n\H^\ast + \log_2(\H^\ast+1)n + \O(\log(\log(\H^\ast+1)+1)n).\]
However, such improvements may soon become negligible in comparison with the overhead of having to compute the value of~$\bt$.

\subsection{Refined upper bounds for \cASS}
\label{subsec:cASS:precise}

\begin{proposition}
\label{pro:cASS:not-tight}
The algorithm \cASS does not have the tight middle-growth property.
\end{proposition}

\begin{proof}
Let~$A_k$ be an array whose run decomposition consists in runs of lengths~$1, 2, 1, 4, 1, 8, 1$, {}$16, 1, \ldots, 1 , 2^k, 1$ for some integer~$k \geqslant 0$.
When sorting the array~$A_k$, the algorithm \cASS keeps merging the two leftmost runs at its disposal.
The resulting tree, represented in Figure~\ref{fig:22-cASS-tree}, has height~$h = 2k$ and its root is a run of length~$r = 2^{k+1}+k-1 = o(2^h)$.
\end{proof}

\begin{figure}[h!]
\begin{center}
\begin{tikzpicture}[scale=0.65]
\draw[very thick]
(0,0) -- (0.5,1) -- (1,0)
(0.5,1) -- (1,2) -- (2,1) -- (2,0);
\foreach \i in {2,...,7}{
 \draw[very thick]
 (\i-1,\i) -- (\i,\i+1) -- (\i+1,\i) -- (\i+1,0);
}
\foreach \i/\j/\k in {0/0/1,1/0/2,2/0/1,3/0/4,4/0/1,5/0/8,6/0/1,%
7/0/16,8/0/1,0.5/1/3,1/2/4,2/3/8,3/4/9,4/5/17,%
5/6/18,6/7/34,7/8/35}{
 \draw[fill=white,very thick] (\i,\j) circle (0.4);
 \node at (\i,\j) {$\k$};
}
\end{tikzpicture}
\end{center}
\caption{Merge tree induced by \cASS on~$A_4$.
Each run is labelled by its length.\label{fig:22-cASS-tree}}
\end{figure}

\begin{theorem}
\label{thm:cASS-log}
Theorems~\ref{thm:PoS-constant} and~\ref{thm:PoS-log} remain valid if we consider the algorithm \cASS instead of an algorithm with the tight middle-growth property.
\end{theorem}

\begin{proof}
Our proof is similar to the proofs of Theorems~\ref{thm:PoS-constant} and~\ref{thm:PoS-log}.
Let~$\T$ be the merge tree induced by \cASS on an array~$A$ of length~$n$
and dual runs~$S_1,S_2,\ldots,S_\sigma$.
Following the terminology of~\cite{Ju20}, we say that a run~$R$ in~$\T$ is \emph{non-expanding} if it has the same level as its parent~$\aR{1}$, i.e., if~$\ell = \al{1}$.
It is shown, in the proof of~\cite[Theorem 3.2]{Ju20}, that the lengths of the non-expanding runs sum up to an integer smaller than~$3n$.
Hence, we partition~$\T$ as follows.

We place all the non-expanding runs into one set~$\R_{\bot}$.
Then, for each~$\ell \geqslant 0$, we define~$\R_\ell$ as the set of expanding runs in~$\T^\ast$ with level~$\ell$, where~$\T^\ast$ is the set of all internal nodes of~$\T$.
By construction, each run~$R$ in~$\R_\ell$ has a length~$r \geqslant 2^\ell$, and the elements of~$\R_\ell$ are pairwise incomparable.

Now, we first observe that
\begin{align*}
\sum_{R \in \R_\bot} \sum_{i=1}^\sigma \cost^\ast_\bt(r_{\rightarrow i}) & \leqslant \sum_{R \in \R_\bot} \sum_{i=1}^\sigma 2 r_{\rightarrow i}
= \sum_{R \in \R_\bot} 2 r \leqslant 6n \text{ and} \\
\sum_{R \in \R_\bot} \sum_{i=1}^\sigma \cost^\ast_{\log}(r,r_{\rightarrow i})
& \leqslant \sum_{R \in \R_\bot} \sum_{i=1}^\sigma r_{\rightarrow i}
= \sum_{R \in \R_\bot} r \leqslant 3n.
\end{align*}
Consequently, like in the proofs of Theorems~\ref{thm:PoS-constant} and~\ref{thm:PoS-log}, we just need to show for all~$i \leqslant \sigma$ that
\begin{align*}
\sum_{\ell \geqslant 0} \text{\rlap{$\C_{\bt}(\ell)$}\phantom{$\C_{\log}(\ell)$}}
& \leqslant (1+1/(\bt+3)) \log_2(n/s_i) s_i + \log_2(\bt+1) s_i + \O(s_i) \text{ and} \\
\sum_{\ell \geqslant 0} \C_{\log}(\ell) & \leqslant \log_2(n/s_i) s_i + 2 \log_2(\log_2(n / s_i) + 1) s_i + \O(s_i),
\end{align*}
where~$\C_{\bt}(\ell) = \sum_{R \in \R_\ell} \cost_{\bt}^\ast(r_{\rightarrow i})$ and~$\C_{\log}(\ell) = \sum_{R \in \R_\ell} \cost_{\log}^\ast(r,r_{\rightarrow i})$.

Note, however, that
\begin{align*}
\text{\rlap{$\C_{\bt}(\ell)$}\phantom{$\C_{\log}(\ell)$}} & \leqslant (1+1/(\bt+3)) \sum_{R \in \R_\ell} r_{\rightarrow i} \leqslant (1+1/(\bt+3)) s_i \text{ and} \\
\C_{\log}(\ell) & \leqslant \sum_{R \in \R_\ell} r_{\rightarrow i} \leqslant s_i
\end{align*}
for all~$\ell \geqslant 0$.
Then, if we set~$\mu = \lceil \log_2((\bt+1)n/s_i\rceil$ and~$\nu = \lceil \log_2(n/s_i \log_2(2 n/s_i)^2) \rceil$, we observe that~$r \geqslant 2^{\mu+h} \geqslant 2^h (\bt+1)n/s_i$ for all~$r \in \R_{\mu+h}$, and that~$r \geqslant 2^{\nu+h} \geqslant 2^h n/s_i \log_2(2 n /s_i)^2$ for all~$r \in \R_{\nu+h}$.
Hence, applying Lemmas~\ref{lem:slice-1} and~\ref{lem:slice-2} proves, as desired, that
\begin{align*}
\sum_{\ell=0}^{\mu-1} \C_{\bt}(\ell) + \sum_{\ell \geqslant 0} \C_{\bt}(\ell+\mu)
& \leqslant (1+1/(\bt+3)) \mu s_i + \O(s_i) \\
& \leqslant (1+1/(\bt+3)) \log_2(n/s_i) s_i + \log_2(\bt+1) s_i + \O(s_i) \text{ and} \\
\sum_{\ell=0}^{\nu-1} \C_{\log}(\ell) + \sum_{\ell \geqslant 0} \C_{\log}(\ell+\nu)
& \leqslant \nu s_i + \O(s_i) \\
& \leqslant \log_2(n/s_i) s_i + 2 \log_2(\log_2(n/s_i)+1) s_i + \O(s_i).
\qedhere
\end{align*}
\end{proof}

\subsection{Refined upper bounds for \PeS}
\label{subsec:PeS:precise}

\begin{proposition}
\label{pro:PeS:not-tight}
The algorithm \PeS does not have the tight middle-growth property.
\end{proposition}

\begin{proof}
Let~$B_k$ be an array whose run decomposition consists in runs of lengths~$1, 1, 3, 3, 9, 9, \ldots, 3^k 3^k$ for some integer~$k \geqslant 0$.
Its length is the integer $b_k = 3^{k+1}-1$ and, when $k \geqslant 1$, it consists in a copy of $B_{k-1}$ and then two runs of lengths $3^k$.
Thus, when sorting $B_k$, \PeS shall recursively sort $B_{k-1}$, then merge it with its neighbouring run of length $3^k$, and merge the resulting run with the rightmost run of length $3^k$.
The resulting tree, represented in Figure~\ref{fig:32-PeS-tree}, has height~$h = 2k+1$ and its root is a run of length~$b_k = 3^{k+1}-1 = o(2^h)$.
\end{proof}

\begin{figure}[h!]
\begin{center}
\begin{tikzpicture}[scale=0.65]
\draw[very thick]
(0,0) -- (0.5,1) -- (1,0)
(0.5,1) -- (1,2) -- (2,1) -- (2,0);
\foreach \i in {2,...,6}{
 \draw[very thick]
 (\i-1,\i) -- (\i,\i+1) -- (\i+1,\i) -- (\i+1,0);
}
\foreach \i/\j/\k in {0/0/1,1/0/1,2/0/3,3/0/3,4/0/9,5/0/9,6/0/27,%
7/0/27,0.5/1/2,1/2/5,2/3/8,3/4/17,4/5/26,%
5/6/53,6/7/80}{
 \draw[fill=white,very thick] (\i,\j) circle (0.4);
 \node at (\i,\j) {$\k$};
}
\end{tikzpicture}
\end{center}
\caption{Merge tree induced by \cASS on~$B_3$.
Each run is labelled by its length.\label{fig:32-PeS-tree}}
\end{figure}

The method of casting aside runs with a limited total length, which we used while adapting the proofs Theorems~\ref{thm:PoS-constant} and~\ref{thm:PoS-log} to \cASS, does not work with the algorithm \PeS.
Instead, we rely on the approach employed by Bayer~\cite{Ba75} to prove the~$n\H+2n$ upper  bound on the number of element comparisons and moves when \PeS uses the naïve merging routine.
This approach is based on the following result, which relies on the notions of \emph{split runs} and \emph{growth rate} of a run.

In what follows, let us recall some notations:
given an array~$A$ of length~$n$ whose run decomposition consists of runs~$R_1,R_2,\ldots,R_\rho$, we set~$e_i = r_1+ \ldots + r_i$ for all integers~$i \leqslant \rho$, the run that results from merging consecutive runs~$R_i,R_{i+1},\ldots,R_j$ is denoted by~$R_{i \ldots j}$.

\begin{definition}
\label{def:split:runs}
Let~$\T$ be a merge tree, and let~$R$ and~$R'$ be the children of a run~$\ovR$.
The \emph{split run} of~$\ovR$ is defined as the rightmost leaf descending from~$R$ if~$r \geqslant r'$, and as the leftmost leaf descending from~$R'$ otherwise.
The length of that split run is then called the \emph{split length} of~$\ovR$, and is denoted by~$\sle(\ovR)$.
Finally, the quantity~$\log_2(\ovr) - \max\{\log_2(r),\log_2(r')\}$ is called \emph{growth rate} of the run~$\ovR$, and is denoted by~$\gr(\ovR)$.
\end{definition}

\begin{lemma}
\label{lem:bayer}
Let~$\T$ be a merge tree induced by \PeS.
We have~$\gr(\ovR) \ovr + 2 \sle(\ovR) \geqslant \ovr$ for each internal node~$\ovR$ of~$\T$.
\end{lemma}

\begin{proof}
Let~$R = R_{i \ldots k}$ and~$R' = R_{k+1 \ldots j}$ be the children of the run~$\ovR$.
We assume, without loss of generality, that~$r \geqslant r'$.
The case~$r < r'$ is entirely symmetric.
Definition~\ref{def:tree:PeS} then states that
\[|r - r'| = |2 e_k - e_j - e_{i-1}| \leqslant
|2 e_{k-1} - e_j - e_{i-1}| = |r - r' - 2 r_k|,\]
which means that~$r- r' - 2 r_k$ is negative
and that~$r - r' \leqslant 2 r_k + r' - r$.

Then, observe that the function~$f \colon t \mapsto 4t-3-\log_2(t)$ is non-negative on the interval~$[1/2,+\infty)$.
Finally, let~$z = r / \ovr$, so that~$z \geqslant 1/2$,~$\gr(\ovR) = -\log_2(z)$ and~$r' = (1-z) \ovr$:
\[\gr(\ovR) \ovr + 2 \sle(\ovR)
= 2 r_k - \log_2(z) \ovr \geqslant
2 (r - r') - \log_2(z) \ovr =
(f(z)+1) \ovr \geqslant \ovr. \qedhere\]
\end{proof}

\begin{lemma}
\label{lem:split-length}
Let~$\T$ be a merge tree induced by \PeS on an array of length~$n$, and let~$\T^\ast$ be the set of internal nodes of~$\T$.
We have~$\sum_{R \in \T^\ast} \sle(R) \leqslant 2n$.
\end{lemma}

\begin{proof}
Let~$\R$ be the set of leaves of~$\T$.
Each run~$R \in \R$ is a split run of at most two nodes of~$\T$: these are the parents of the least ancestor of~$R$ that is a left run (this least ancestor can be~$R$ itself) and of the least ancestor of~$R$ that is a right run.
It follows that
\[\sum_{R \in \T^\ast} \sle(R) \leqslant 2 \sum_{R \in \R} r = 2 n. \qedhere\]
\end{proof}

\begin{lemma}
\label{lem:small-growth}
Let~$\T$ be a merge tree induced by \PeS on an array with~$\sigma$ dual runs~$S_1,S_2,\ldots,S_\sigma$.
For all~$i \leqslant \sigma$ and all sets~$\X$ of inner nodes of~$\T$, we have
\[\sum_{R \in \X} \gr(R) r_{\rightarrow i} \leqslant \log_2(x) s_i,\]
where~$x$ is the largest length of a run~$R \in \X$.
\end{lemma}

\begin{proof}
Let~$\R$ be the set of leaves of~$\T$.
Without loss of generality, we assume that~$\X$ contains each node $\ovR$ of~$\T$ with length~$\ovr \leqslant x$.
Now, for all runs~$R \in \R$, let~$R^\uparrow$ be the set of those strict ancestors of~$R$ that belong to~$\X$.
Since every child of a node in~$\X$ is either a leaf or a node in~$\X$, the set~$R^\uparrow$ consists of those runs~$\aR{k}$ such that~$1 \leqslant k \leqslant |R^\uparrow|$.
It follows that
\begin{align*}
\sum_{\ovR \in \X} \gr(\ovR) \ovr_{\rightarrow i} &
= \sum_{R \in \R} \sum_{\ovR \in R^\uparrow} \gr(\ovR) r_{\rightarrow i}
= \sum_{R \in \R} \sum_{k=1}^{|R^\uparrow|} \gr(\aR{k}) r_{\rightarrow i} \\
& \leqslant \sum_{R \in \R} \sum_{k=1}^{|R^\uparrow|} \log_2(\ar{k} / \ar{k-1}) r_{\rightarrow i}
= \sum_{R \in \R} \log_2(\ar{|R^\uparrow|} / r) r_{\rightarrow i} \\
& \leqslant \sum_{R \in \R} \log_2(x) r_{\rightarrow i} = \log_2(x) s_i.
\qedhere
\end{align*}
\end{proof}

\begin{theorem}
\label{thm:PeS-constant}
Theorem~\ref{thm:PoS-constant} remains valid if we consider the algorithm \PeS instead of an algorithm with the tight middle-growth property.
\end{theorem}

\begin{proof}
Let~$\T$ be the merge tree induced by \PeS on an array~$A$ of length~$n$ and dual runs~$S_1,S_2,\ldots,S_\sigma$, and let~$\T^\ast$ be the set of internal nodes of~$\T$.
The algorithm \PeS requires no more than
\[\C = \sum_{R \in \T^\ast} \sum_{i=1}^\sigma \cost_\bt^\ast(r_{\rightarrow i}) + \O(n)\]
comparisons and, for a given run~$R \in \T^\ast$, we have
\begin{align*}
\sum_{i=1}^\sigma \cost_\bt^\ast(r_{\rightarrow i}) - 4 \sle(R)
& {{}\leqslant{}} \sum_{i=1}^\sigma \cost_\bt^\ast(r_{\rightarrow i}) - 2 (1+1/(\bt+3)) \sle(R) \\
& {{}\leqslant{}} \sum_{i=1}^\sigma \min\{(1+1/(\bt+3)) r_{\rightarrow i},\bt+2+2\log_2(r_{\rightarrow i}+1)\} + \phantom{.} \\
& \phantom{{}\leqslant{}} \phantom{\smash{\sum_{i=1}^\sigma}~´} (1+1/(\bt+3)) (\gr(R)-1) r_{\rightarrow i} \\
& {{}\leqslant{}} \sum_{i=1}^\sigma \cost_\bt^{\ast\ast}(\gr(R),r_{\rightarrow i}),
\end{align*}
where we set~$\cost_\bt^{\ast\ast}(\gamma,m) = \min\{(1+1/(\bt+3)) \gamma m,\bt+2+2\log_2(m+1)\}$.

Since \PeS has the fast-growth property, there exist a real number~$\alpha > 1$ and an integer~$\delta \geqslant 0$ such that~$\ar{k} \geqslant \alpha^{k-\delta} r$ for all runs~$R$ of depth at least~$k$ in~$\T$.
Now, given an index~$i \leqslant \sigma$, let~$\X$ be the set of runs~$R$ of length~$r < \alpha^\delta (\bt+1)n/s_i$, and let~$\Y$ be the sub-tree of~$\T$ consisting of runs of length~$r \geqslant \alpha^\delta(\bt+1)n/s_i$.
For all~$h \geqslant 0$, let~$\Y_h$ denote the set of runs of height~$h$ in the tree~$\Y$; by construction,~$\Y_h$ contains only runs of length~$r \geqslant \alpha^h (\bt+1)n/s_i$.

Consequently, Lemma~\ref{lem:small-growth} states that
\[\sum_{R \in \X} \cost_\bt^{\ast\ast}(\gr(R),r_{\rightarrow i}) \leqslant (1+1/(\bt+3)) \sum_{R \in \X} \gr(R) r_{\rightarrow i} \leqslant (1+1/(\bt+3)) \log_2(\alpha^\delta (\bt+1)n/s_i) s_i,\]
whereas Lemma~\ref{lem:slice-1} states that
\[\sum_{R \in \Y} \cost_\bt^{\ast\ast}(\gr(R),r_{\rightarrow i}) \leqslant 
\sum_{R \in \Y} \cost_\bt^\ast(r_{\rightarrow i}) \leqslant 9 \alpha s_i / (\alpha-1).\]
It follows that
\begin{align*}
\C & \leqslant \sum_{R \in \T^\ast} 4 \sle(R) + \sum_{i=1}^\sigma \cost_\bt^{\ast\ast}(\gr(R),r_{\rightarrow i}) + \O(n) \\
& \leqslant \sum_{i=1}^\sigma \Big((1+1/(\bt+3)) \log_2(\alpha^\delta (\bt+1) n/s_i) s_i + \O(s_i)\Big) + \O(n) \\
& \leqslant (1+1/(\bt+3)) n \H^\ast + n \log_2(\bt+1) + \O(n).
\qedhere
\end{align*}
\end{proof}

\begin{theorem}
\label{thm:PeS-log}
Theorem~\ref{thm:PoS-log} remains valid if we consider the algorithm \PeS instead of an algorithm with the tight middle-growth property.
\end{theorem}

\begin{proof}
Let us reuse the notations and the main lines of the proof of Theorem~\ref{thm:PeS-constant}.
With these notations, \PeS requires no more than
\[\C = \sum_{R \in \T^\ast} \sum_{i=1}^\sigma \cost^\ast_{\log}(r,r_{\rightarrow i}) + \O(n)\]
element comparisons and, for a given run $R \in \T^\ast$, we have
\begin{align*}
\sum_{i=1}^\sigma \cost^\ast_{\log}(r,r_{\rightarrow i}) - 2 \sle(R)
& \leqslant \sum_{i=1}^\sigma \min\{r_{\rightarrow i},6\log_2(r+1)^2+6\} + (1-\gr(R)) r_{\rightarrow i} \\
& \leqslant \sum_{i=1}^\sigma \cost^{\ast\ast}_{\log}(r,\gr(R) r_{\rightarrow i}),
\end{align*}
where we set $\cost^{\ast\ast}_{\log}(r,m) = \min\{m,6\log_2(r+1)^2+6\}$.

Let $\alpha > 1$ and $\delta \geqslant 0$ be a real number and an integer such that $\ar{k} \geqslant \alpha^{k - \delta} r$ for all runs $R$ of depth at least $k$ in $\T$.
Then, given an integer $i \leqslant \sigma$, let~$\X$ be the set of runs~$R$ of length~$r < \alpha^\delta n/s_i \log_2(2 n/s_i)^2$, and let~$\Y$ be the sub-tree of~$\T$ consisting of runs of length~$r \geqslant \alpha^\delta n/s_i \log_2(2 n/s_i)^2$.
For all~$h \geqslant 0$, let~$\Y_h$ denote the set of runs of height~$h$ in the tree~$\Y$; by construction,~$\Y_h$ contains only runs of length~$r \geqslant \alpha^h n/s_i \log_2(2 n/s_i)^2$.

Consequently, Lemma~\ref{lem:small-growth} states that
\[\sum_{R \in \X} \cost_{\log}^{\ast\ast}(r,\gr(R)r_{\rightarrow i}) \leqslant \sum_{R \in \X} \gr(R) r_{\rightarrow i} \leqslant \log_2(\alpha^\delta n/s_i \log_2(2 n/s_i)^2) s_i,\]
whereas Lemma~\ref{lem:slice-1} states that
\[\sum_{R \in \Y} \cost_{\log}^{\ast\ast}(r,\gr(R)r_{\rightarrow i}) \leqslant 
\sum_{R \in \Y} \cost_{\log}^\ast(r,r_{\rightarrow i}) = \O(s_i).\]
It follows that
\begin{align*}
\C & \leqslant \sum_{R \in \T^\ast} 2 \sle(R) + \sum_{i=1}^\sigma \cost_{\log}^{\ast\ast}(r,\gr(R)r_{\rightarrow i}) + \O(n) \\
& \leqslant \sum_{i=1}^\sigma \Big(\log_2(\alpha^\delta n/s_i \log_2(2 n/s_i)^2) s_i + \O(s_i)\Big) + \O(n) \\
& \leqslant n \H^\ast + 2 n \log_2(\H^\ast+1) + \O(n).
\qedhere
\end{align*}
\end{proof}

\end{document}